\documentclass[twoside,12pt]{article}
\usepackage[margin=3cm]{geometry}
\usepackage{enumitem}

\usepackage{amsmath,amsfonts,bm}

\def\eqref#1{equation~\ref{#1}}

\def\1{\bm{1}}

\DeclareMathAlphabet{\mathsfit}{\encodingdefault}{\sfdefault}{m}{sl}
\SetMathAlphabet{\mathsfit}{bold}{\encodingdefault}{\sfdefault}{bx}{n}

\usepackage{hyperref}

\usepackage{url}            
\usepackage{booktabs}       
\usepackage{nicefrac}       
\usepackage{microtype}      
\usepackage{adjustbox}
\usepackage[flushleft]{threeparttable}
\usepackage{wrapfig}
\usepackage{tcolorbox}
\usepackage[page, header]{appendix}
\usepackage[utf8]{inputenc} 
\usepackage[T1]{fontenc}    
\usepackage{url}            
\usepackage{booktabs}       
\usepackage{caption}
\usepackage{amsfonts}       
\usepackage{nicefrac}       
\usepackage{microtype}      
\usepackage{bm}
\usepackage{diagbox}
\usepackage{color} 
\usepackage{algorithmic}
\usepackage{amssymb}
\usepackage{amsmath}
\usepackage{amsfonts}
\usepackage{enumitem}

\usepackage{mathtools}
\usepackage{lineno}
\usepackage{minitoc}
\doparttoc 
\faketableofcontents 
\usepackage[labelformat=simple]{subcaption}
\usepackage{wrapfig}

\makeatother
\newtheorem{proof}{Proof}

\newtheorem{theorem}{Theorem}
\newtheorem{definition}{Definition}
\newtheorem{proposition}{Proposition}
\newtheorem{corollary}{Corollary}
\newtheorem{lemma}{Lemma}

\usepackage{natbib}
\usepackage{multicol, multirow}
\usepackage[ruled,vlined]{algorithm2e}

\newcommand{\revision}[1]{{\color{black} #1}} 

\newcommand{\mc}{\mathcal}

\newcommand{\cF}{\mc{F}}

\newcommand{\cA}{{\mathcal A}}
\newcommand{\cB}{{\mathcal B}}
\newcommand{\cC}{{\mathcal C}}

\newcommand{\cG}{{\mathcal G}}

\newcommand{\cP}{{\mathcal P}} 
 
\newcommand{\cR}{{\mathcal R}} 
\newcommand{\cS}{{\mathcal S}} 

\newcommand{\cM}{{\mathcal M}}

\newcommand{\cX}{{\mathcal X}}
\newcommand{\cZ}{{\mathcal Z}}

\newcommand{\bma}{{\bm a}}

\newcommand{\bmv}{{\bm v}}

\newcommand{\bmx}{{\bm x}}

\newcommand{\bmQ}{{\bm Q}}

\newcommand{\cN}{\mathcal N}

\SetKwInput{KwInput}{Input}                
\SetKwInput{KwOutput}{Output}

\providecommand{\customgenericname}{}
\newcommand{\newcustomtheorem}[2]{%
  \newenvironment{#1}[1]
  {%
  \renewcommand\customgenericname{#2}%
  \renewcommand\theinnercustomgeneric{##1}%
  \innercustomgeneric
  }
  {\endinnercustomgeneric}
}
\newcustomtheorem{customtheorem}{Theorem}

\usepackage{hyperref}
\hypersetup{
    colorlinks=true,
    linkcolor=red,
    citecolor=blue,
    filecolor=magenta,      
    urlcolor=blue,
linktocpage}
\usepackage{cleveref}

\title{Stochastic Games with Minimally Bounded Action Costs}

\author{
David Mguni\textsuperscript{1}
\\
\textsuperscript{1}Queen Mary University, London\\
} 
\begin{document}

\maketitle

\begin{abstract}
In many
multi-player interactions, 
players incur strictly positive costs each time they execute actions e.g. `menu costs' or transaction costs in financial systems. 
Since acting at each available opportunity would accumulate prohibitively large costs, the resulting decision problem is one in which players must make strategic decisions about \textit{when} to execute actions in addition to their choice of action. 
This paper analyses a discrete-time stochastic game (SG) in which players face minimally bounded positive costs for each action and influence the system using impulse controls. 
%
%
We prove SGs of two-sided impulse control have a unique value and characterise the saddle point equilibrium in which the players execute actions at strategically chosen times in accordance with Markovian strategies. We prove the game respects a dynamic programming principle and that the Markov perfect equilibrium can be computed as a limit point of a sequence of Bellman operations. We then introduce a new Q-learning variant which we show converges almost surely to the value of the game enabling solutions to be extracted in unknown settings.   
%
%
Lastly, we extend our results to settings with
budgetory constraints. \end{abstract}

\section{Introduction}

The goal of successfully modelling financial behaviour is the cornerstone of financial theory. Although this ambition holds the key to greater mastery over financial systems, progress towards it is hindered by the fact that various economic intricacies remain neglected in many of today's financial models \citep{waggoner2012confronting,hansen2001acknowledging}. Model misspecifications of this kind can result in vast disparities between the predictions of financial models and the financial phenomena they attempt to model \citep{simons1997model,krugman2009did}. Correctly modelling economic behaviour requires accurately specifying the costs and rewards associated with economic activity in a given system.

Stochastic games model strategic interactions between two competing agents that take place over time \citep{shapley1953stochastic,solan2015stochastic}. 
They are a standard modelling framework for competitive settings within economics and finance  \citep{prasad2004competitive,browne2000stochastic}. Nevertheless, classical stochastic games do not naturally induce any restrictions on the cost magnitudes for taking actions. Therefore, in these models, costs can be made arbitrarily small given specific choices of action. Consequently, agents can freely execute actions at each available opportunity without accumulating large costs over time. In many financial settings, despite agents having the ability to exercise control over the magnitude of their investments, the cost of agents' investment decisions cannot be made arbitrarily small~\citep{cuypers2021transaction}. For example, the amount of stock that an agent can purchase may be required to be at least some fixed minimal amount e.g. an individual share unit or, agents may be required to pay transaction costs i.e., a payment to a broker to perform their investment decisions on their behalf~\citep{allen1991transaction,goldstein2009brokerage}.\footnote{Fixed costs of this kind are also prevalent in microeconomic settings where firms that make adjustments to their business strategies can be subject to adjustment costs or \textit{menu costs}~\citep{midrigan2011menu}.} 
  %
In these settings, investment strategies that prescribe performing actions at every opportunity can be vastly suboptimal since acting in this way can accumulate massive costs~\citep{azimzadeh2017impulse}. Consequently, in settings where agents face minimally bounded costs for their actions, models derived from classical stochastic games may fail to generate predictions that accord with observed outcomes. 

To resolve this issue in single-agent settings, a form of policy known as \textit{impulse control} has emerged within optimal control theory~\citep{oksendal2005applied,mguni2022timing,davis2010impulse}. Impulse control models prescribe a set of optimal points to perform actions in addition to the sequence of optimal best-response actions \citep{azimzadeh2017impulse,jeanblanc1993impulse}. Currently however, the theory of stochastic games in discrete-time does not include a treatment that handles impulse control. As such, devising algorithmic protocols for computing equilibrium strategies in stochastic games with minimally bounded action costs is largely unfeasible.  
To address this issue, in this paper, we tackle this problem by studying a stochastic game in which the players face minimally bounded costs for each action. Our treatment includes an algorithmic learning framework for computing the minimax equilibrium strategies. To our knowledge, this treatment is the first to tackle learning in dynamic strategic environments of this kind.

\textbf{Example: Advertising Investment Duopoly Costs~\citep{deal1979optimizing,prasad2004competitive}.} To demonstrate the ideas within a concrete setting, we consider a well-known problem within economics, namely the problem of how a firm should invest in advertising in a duopoly setting in which it competes for market share. In this setting, each firm seeks to maximise its long-term profit and to this end, performs investment adjustments.  
At time $t=0,1,\ldots$ each Firm $i\in\{1,2\}$ has a revenue stream $S^i_t=S^i(\omega): \mathbb{R}_{>0}\times \Omega\to \mathbb{R} $ which is a stochastic process. At any point, Firm $i$ makes costly investments of size $u_i\in \mathcal{U}_i$ where $\mathcal{U}_i$ denotes the set of admissible investments for Firm $i$. Denote by $M\in \mathbb{R}_{>0}$ the potential market size and by $b^i\in ]0,1]$ the response rate to advertising for Firm $i$, then the revenue stream for Firm $i$ is given by the following expression:
\begin{align}
S^i_{t+1}=S^i_t+b^i u^i_t\left[M- S^i_t- S^j_t\right] M ^{-1}  - r^iS^i_t
+\sigma^i\left(B_{t+1}-B_t\right),\;\label{marketshareeqn1}
\end{align}
where $i\neq j\in \{1,2\}$ and  $s^i\equiv S^i_0 \in \mathbb{R}_{>0}$ are the initial sales for Firm $i$. The constants $r^i,\sigma^i\in \mathbb{R}$ represent the rate at which Firm $i$ abstracts market share and the volatility of the sales process for Firm $i$ respectively. The term $B_t$ is Brownian motion which introduces randomness to the system. 
Each Firm $i$ seeks to maximise its cumulative profit which consists of its revenue due to sales $h^i:\mathbb{R}\to \mathbb{R}$ minus its running advertising costs $c_i:\mathbb{R}_{>0}\times \mathbb{R}\to \mathbb{R}$. The profit function for Firm $i$, $\Pi^i$ is given by the following expression:
\begin{align} 
\Pi^i (s^i;u^i,u^j )=\mathbb{E}\left[\sum_{t\geq0}\gamma^t\left(h^i(S^i_t)-\left[c^i(u^i_t)-c^j(u^j_t)\right]\right)\right]. 
\label{profitfunctioncontgame}
\end{align}
Since the market is duopolistic, the payoff structure satisfies $\Pi_1 +\Pi_2 =0$.
%
Therefore the firms engage in a zero-sum stochastic game.  Models of this kind have been used to analyse the strategic interaction within advertising duopoly using a game-theoretic framework \citep{prasad2004competitive}. Using this framework, the behaviour of the firms in the advertising problem can be characterised by computing the best-response strategies within the stochastic game. Models such as the {Vidale-Wolfe model of advertising} \cite{sethi1973optimal} have been studied to analyse settings in which both firms make continuous modifications to their investment positions models. These models do not account for the fixed minimal costs that firms incur when executing investment decisions. We refer the reader to \citep{erickson1995differential} for exhaustive discussions on duopoly advertising investment models and to \citep{prasad2004competitive} for a stochastic differential game approach. 

In this paper, we study a stochastic game that models the behaviour of firms (and other economic agents) facing minimally-bounded fixed costs for each investment. This feature prohibits firms from committing to strategies that involve performing (possibly infinitesimal) adjustments over the horizon of the problem contrary to current strategic investment models. 
%
To handle scenarios in which the game model may not be known in advance, we develop a reinforcement learning framework \citep{sutton2018reinforcement, marl-book} that enables the equilibrium strategies to be computed using repeated interactions with the environment. In particular, we modify the standard reinforcement learning paradigm in which agents solely learn to optimally select correct actions at each time-step  \citep{szepesvari2022algorithms} to handle strategic settings in which the agents must also learn to determine when to perform actions. 
Our learning protocol enables the equilibrium strategies for the game to be computed using a value-iterative procedure. 
%
%
    %
 We then establish results that ensure convergence of a Q-learning variant to the minimax equilibrium of the game.  To do this, we prove a series of results namely:

\noindent\textbf{i)} We prove the existence of a saddle point equilibrium of a (discrete-time) Markov game in which players use impulse control and, therefore strategically select the set of points to perform actions. 

\noindent\textbf{ii)} We then establish a dynamic programming principle (DPP) for stochastic games of two-sided impulse control and show that the minimax (saddle point) equilibrium can be computed as a limit point of a sequence of Bellman operations (Theorem \ref{theorem:joint-sol}) which lays the foundation for a learning methodology for computing minimax equilibrium solutions to SGs of two-sided impulse control.
\newline\textbf{iii)} We extend result ii) to a new variant of Q learning which enables the game to be solved even when the game is unknown using a multi-agent reinforcement learning method (Theorem \ref{theorem:q_learning}).    
\newline\textbf{iv)} We then extend the result i) to (linear) function approximators enabling the value function to be parameterised (Theorem \ref{primal_convergence_theorem}). \newline\textbf{v)} In Sec. \ref{sec:budget}, we extend our analysis to include budgetary constraints so that each action draws from a fixed budget which each player must stay within. Analogous to the development of ii), we establish another DPP from which we derive a Q-learning variant for tackling impulse control with budgetary constraints (Theorem \ref{thm:maxmin_budget}). A particular case of a budget constraint is when the number of actions each player can take over the horizon is capped. 






\section{Formulation}
\textbf{Stochastic Games.} In a (two-player) stochastic game (SG) \citep{shapley1953stochastic}, a pair of players sequentially perform actions to maximise their individual expected returns. The underlying problem is formalised by a tuple $\left\langle \cN, \mathcal{S},\cA,\cB,P,R,\gamma\right\rangle$ where $\mathcal{S}\subset \mathbb{R}^p$ is the set of states, $\cA\subset \mathbb{R}^k$ and $\cB\subset \mathbb{R}^k$ are the set of actions for Player 1 and Player 2 respectively, $P:\mathcal{S} \times \cA\times \cB \times \mathcal{S} \rightarrow [0, 1]$ is a transition probability function describing the system's dynamics, $R: \mathcal{S} \times \cA\times \cB \rightarrow \mathbb{R}$ is the reward function and the discount factor $\gamma \in [0, 1)$ specifies the degree to which each player's rewards are discounted over time. At time $t\in 0,1,\ldots, $ the system is in state $s_{t} \in \mathcal{S}$ and Player 1 and Player 2 must
choose an action,  $a_t \in \cA$  and $b_t \in \cB$ respectively which transitions the system to a new state 
$s_{t+1} \sim P(\cdot|s_t, a_t, b_t)$ and produces a reward $R(s_t, a_t, b_t)$. The goal of Player $1$ is to
determine a strategy $\sigma^1\in\Sigma^1$ that maximises its expected return given by the value function: $
v^{\sigma^1,\sigma^2}(s)=\mathbb{E}[\sum_{t=0}^\infty \gamma^tR(s_t,a_t,b_t)|a_t\sim\sigma^1(\cdot|s_t),b_t\sim\sigma^2(\cdot|s_t), s_0=s]$ whereas the goal of Player 2 is to minimise the same quantity using a strategy $\sigma^2\in \Sigma^2$ where $\Sigma^i$ is the policy set for player $i\in\{1,2\}$.  The action value function is given by $\bm{Q}(s,a,b)=\mathbb{E}[\sum_{t=0}^\infty R(s_t,a_t,b_t)|a_0=a,b_0=b, s_0=s] $. 
A Markov policy $\sigma: \mathcal{S} \times \mathcal{A} \rightarrow [0,1]$ for  Player 1  is a probability distribution over state-action pairs where $\sigma(a|s)$ represents the probability of selecting action $a\in\mathcal{A}$ in state $s\in\mathcal{S}$. We construct a Markov strategy for Player 2 analogously.  In this paper, we restrict our attention to Markov policies for both players. 


We consider a setting in which each player faces at least some minimal cost for each action it performs. Systems of this kind widely occur in economic and financial systems. Common examples are duopolies in which firms face investment adjustment costs. With this, the objective that Player $1$ seeks to maximise is given by:
\begin{align}
v^{\sigma^1,\sigma^2}(s)=\mathbb{E}\left[\sum_{t=0}^\infty \gamma^t\left\{\cR(s_t, a_t,b_t)- \cC(s_t, a_t,b_t)\right\}\Big|s_0=s\right], \label{impulse_objective}
\end{align}
and the goal of Player $2$ is to maximise $-v$, where for any state $s\in\cS$ and any action $a\in\cA,b\in\cB$, the function $\cR$ is given by $
\cR(s,a,b):=R(s,a\boldsymbol{1}_{\cA}(a),b\boldsymbol{1}_{\cB}(b)), $
where $\boldsymbol{1}_{\mathcal{Y}}(y)$ is the indicator function which is $1$ whenever $y\in\mathcal{Y}$ and $0$ otherwise. For example, if at time $t$ only Player 1 takes an action $\hat{a}^1_t\in\cA_i$ while player $2$ takes no action, the reward is $R(s_t,\hat{a}_t,0)$.
The cost function $\cC:\cS\times\cA\times\cB\to \mathbb{R}$ is given by $
\cC(s,a,b):=
			c(s,a)\mathbf{1}_{a\in\cA}+c(s,b)\mathbf{1}_{b\in\cB}$
			where $c:\cS\times\left(\cA\cup\cB\right)\to\mathbb{R}_{>0}$ is a minimally bounded       (cost) function\footnote{I.e. a function which is bounded below by a positive constant.} that introduces a cost each time the player performs an action. Examples of the cost function is a quasi-linear function of the form $c(s_t,y_t)=\kappa+f(y_t)$ for any $y_t\in\cA\cup\cB$, where $f:\cA\cup\cB\to\mathbb{R}_{>0}$ and $\kappa$ is a real-valued positive constant. In this game, at each time step $t$, the system transitions according to the probability kernel $\cP$ which is given by $
\cP(s_{t+1},a_t,b_t,s_t):=P(s_{t+1},a_t\boldsymbol{1}_{\mathcal{A}}(a_t),b_t\boldsymbol{1}_{\mathcal{B}}(b_t),s_t)$.
Since acting at each time step would incur prohibitively high costs, the players must be selective about when to act. Therefore, in this setting, the player's problem is augmented to deciding both a best-response strategy for their actions and when to apply their chosen action. The question of when each player ought to act is one of the key questions addressed in this paper.

The function $c$ plays an important role in modulating the willingness of each player to act. Larger costs for each action mean the players must be more selective about executing actions, reserving their actions to a smaller number of states that induce larger increases in their expected return (this is proven in Theorem \ref{thrm:minimax-exist}).  In the limit $c\to 0$ we return to a classic SG framework where each player is willing to execute an action at each state.  In Sec. \ref{sec:budget}, as an alternative to using the function $c$ as a way of controlling the number of actions executed by the player. There we discuss imposing a budget constraint on the number of actions each player can perform over a given horizon.

\begin{definition}
Let us define by ${\rm val}^+[v]:=\min\limits_{\sigma^2\in\Sigma^2}\max\limits_{\sigma^1\in\Sigma^1}v^{\sigma^1,\sigma^2}$ the \textit{upper value function} and by ${\rm val}^-[v]:=\underset{\sigma^1\in\Sigma^1}{\max}\underset{\sigma^2\in\Sigma^1}{\min}\;v^{\sigma^1,\sigma^2}$, the \textit{lower value function}. The upper (lower) value function represents the minimum payoff that Player $1$ (Player $2$) can guarantee itself irrespective of the actions of the opponent.  The \textbf{value} of the game exists if we can commute the $\max$ and $\min$ operators:
\begin{align}
  {\rm val}^-[v]&=\underset{\sigma^1\in\Sigma^1}{\max}\min_{\sigma^2\in\Sigma^2}v^{\sigma^1,\sigma^2}
    =\min_{\sigma^2\in\Sigma^2}\underset{\sigma^1\in\Sigma^1}{\max}\;v^{\sigma^1,\sigma^2}={\rm val}^+[v]. \label{value_equation}
\end{align}
\end{definition}
We denote the value by $v:={\rm val}^+[v]={\rm val}^-[v]$ and denote by $(\hat{\sigma}^1,\hat{\sigma}^2)\in\Sigma^1   \times\Sigma^2$ the pair that satisfies the equality \eqref{value_equation}. 

In general, the functions ${\rm val}^+[v]$ and ${\rm val}^-[v]$ may not coincide. The value, should it exist, is the minimum payoff each player can guarantee themselves under the equilibrium strategy. If a value exists in Markov strategies, it constitutes a Markov perfect equilibrium \citep{deng2021complexity} of the game in which neither player can improve their payoff by playing some other control --- an analogous concept to a Nash equilibrium for the case of two-player zero-sum games. Thus the central task to establish an equilibrium involves unambiguously assigning a value to the game, that is proving the equality in \eqref{value_equation}.

\section{Existence of a Value and Saddle Point Equilibrium} \label{sec:analysis}


In this section, we prove the existence of a minimax equilibrium value of the game $\cG$. We then characterise the equilibrium conditions for executing actions and show that such times are characterised by an `obstacle condition' which can be evaluated online. We perform some further studies on the properties of the equilibrium strategies. We then show that the dynamic programming principle holds and prove the convergence of a dynamic programming method to the minimax equilibrium of the game.  We then extend the result to allow for (linear) function approximators.
Our first main result is the existence of a unique minimiax equilibrium of the game $\cG$ in Markov strategies (Theorem \ref{thrm:minimax-exist}). Our second key result is that the game respects a dynamic programming principle and that the value of the game can be obtained by computing the sequence of Bellman operations acting on some function. With this, each player enacts a minimax best response policy. 
 
 The results are achieved through several steps: Theorem \ref{theorem:joint-sol} proves that Bellman dynamic programming principle holds and that the Bellman operator for the game is a contraction mapping. This paves the way for proving the convergence of the sequence of repeated application of Bellman operators to the value of the game. Prop. \ref{prop:mpe_result} then proves that the equilibrium strategies derived in Theorem \ref{theorem:joint-sol} consists of Markov strategies and thus the equilibrium concept corresponds to a Markov perfect equilibrium. Thereafter we prove that the value of the game is unique. Thereafter, we characterise the conditions under which each player takes an action when executing its best-response minimax equilibrium strategy. Prop. \ref{prop:switching_times} characterises the conditions under which a rational player should execute an action in response to its adversary and does so in terms of a condition on the action-value function that can be evaluated online. The results of this section lay the foundation for learning methods studied in Sec. \ref{sec:learning} where we consider settings in the game is unknown to the players who seek to determine their equilibrium strategies through repeated interaction. The proof of the results in this section are deferred to the Appendix.
 
\noindent\textbf{Notation \& Assumptions}
 All results are built under Assumptions A.1 - A.5 which are standard in RL and stochastic approximation theory \citep{bertsekas2012approximate}.

We assume that $\mathcal{S}$ is defined on a probability space $(\Omega,\mathcal{F},\mathbb{P})$ and any $s\in\mathcal{S}$ is measurable with respect
to the Borel $\sigma$-algebra associated with $\mathbb{R}^p$. We denote the $\sigma$-algebra of events generated by $\{s_t\}_{t\geq 0}$
by $\mathcal{F}_t\subset \mathcal{F}$. In what follows, we denote by $\left( \mathcal{V},\|\|\right)$ any finite normed vector space and by $\mathcal{H}$ the set of all measurable functions. In what follows, we employ the shorthand $\cP^{\bma}_{s's}=:\sum_{s'\in\cS}P(s';\bma,s)$ for any $s\in \cS$ and for any $\bma\in\bm{\cA}$.
The results of the paper are built under the following assumptions which are standard within stochastic approximation methods:

\textbf{A.1.}
The stochastic process governing the system dynamics is ergodic, that is the process is stationary and every invariant random variable of $\{s_t\}_{t\geq 0}$ is equal to a constant with probability $1$.

\textbf{A.2.}
The function $R$ is in $L_2$.

\textbf{A.3.}
For any positive scalar $c$, there exists a scalar $\kappa_c$ such that for all $s\in\mathcal{S}$ and for any $t\in\mathbb{N}$ we have: $
    \mathbb{E}\left[1+\|s_t\|^c|s_0=s\right]\leq \kappa_c(1+\|s\|^c)$.

\textbf{A.4.}
There exists scalars $C_1$ and $c_1$ such that for any function $v$ satisfying $|v(s)|\leq C_2(1+\|s\|^{c_2})$ for some scalars $c_2$ and $C_2$ we have that: $
    \sum_{t=0}^\infty\left|\mathbb{E}\left[v(s_t)|s_0=s\right]-\mathbb{E}[v(s_0)]\right|\leq C_1C_2(1+\|s_0\|^{c_1c_2})$.

\textbf{A.5.}
There exists scalars $c$ and $C$ such that for any $s\in\cS$ we have that $
    |R(s,\cdot)|\leq C(1+\|s\|^c)$.
    
\begin{definition}
{Given a function $\bmQ:\mathcal{S}\times\cA\cup\cB\to\mathbb{R},\;\forall\sigma^i\in\Sigma^i$ and $\forall s_{\tau},s_\rho\in\mathcal{S}$, we define the intervention operators $\bm\cM_1$ and $\bm\cM_2$ by \\ $
\bm\cM_1\bmQ(s_{\tau},a_{\tau},b):=\max\limits_{a_\tau\in\cA}\left[\cR\left(s_{\tau},a_{\tau},b\right)-c(s_{\tau},a_{\tau})+\gamma\sum_{s'\in\mathcal{S}}P\left(s';a_{\tau},b,s\right)v(s')|\tau\in\cF\right]$, $\forall b \in\cB$,
and \\
$\bm\cM_2\bmQ (s_{\rho},a,b_{\rho}):=\min\limits_{b_\rho\in\cB}\left[\cR\left(s_{\rho},a, b_{\rho}\right)+c(s_{\rho},b_{\rho})+\gamma\sum_{s'\in\mathcal{S}}P\left(s';a,b_{\rho},s\right)v(s')|\rho\in\cF\right]$, $\forall a \in\cA$,
where $\tau\in\cF$ and $\rho\in\cF$ are Player 1 and Player 2 intervention times respectively.}
\end{definition}

%
%
%

The interpretation of $\bm\cM_1$ is the following: suppose that Player 1 is using the policy $\sigma^1$ and at time $t=\tau$ the system is at a state $s_{\tau}$ and the player performs an action $a_{\tau}\sim  \sigma^1(\cdot|s_{\tau})$. A cost of $c(s_{\tau},a_{\tau})$ is then incurred by the player, and the system transitions to $s'\sim P(\cdot;a_{\tau},b,s_{\tau})$. Lastly, recall $v^{\bm\sigma^1,\bm\sigma^2}$ is the value function under the policy pair $(\bm\sigma^1,\bm\sigma^2)$. Therefore, the quantity $\bm\cM_1\bmQ^{\bm\sigma^1,\bm\sigma^2}$ measures the expected future stream of rewards after an immediate action minus the cost of action. This object plays a crucial role which as we later discuss and enables us to characterise the points at which each player should perform an action.

{For any $a\in\cA, b\in\cB$, define $Q_1(s,a):=\bm{Q}(s,a,0)$ and $Q_2(s,b):=\bm{Q}(s,0,b)$, given a function $v:\cS\to \mathbb{R}$, for any $s \in\cS$, we define the Bellman operator $T$,  by: 
%
\begin{align}\label{bellman_op} 
T v(s):=\min\Bigg[\max\left\{\bm\cM_1Q_1, \cR(s,\bm0)+\gamma\sum_{s'\in\mathcal{S}}\cP(s';\bm0,s)v(s')\right\}
,\bm\cM_2Q_2\Bigg].
\end{align}}
%

The Bellman operator captures the nested sequential structure of the decision process.\footnote{Note that the Bellman operator in \eqref{bellman_op} encodes an order of precedence for the players --- if both players decide to perform an action on the system at the same time, we only take into account the action of Player 2. This introduces an obvious asymmetry in the game. Nevertheless, in infinite horizon problems, such overlaps may be rare and hence the effect of this asymmetry is expectedly small.} In particular, the structure in \eqref{bellman_op} consists of an inner structure that consists of two terms: the first term is the expected future return given an action is taken at the current state under the policy $\sigma^i$. The second term is the expected future return given no action is taken at the current state. Lastly, the outer structure is an optimisation that compares the expected return of the two possibilities and selects the maximum.  

\begin{theorem}\label{thrm:minimax-exist}
The minimax value of the game exists and is unique, that is there exists a function $\boldsymbol{\hat{v}}:\cS\to\mathbb{R}$ which respects the following equality: 
\begin{align}
\bm\hat{v}=\underset{\boldsymbol{\hat{\sigma}}^1\in\boldsymbol{\hat{\Sigma}}^1}{\min}\underset{\boldsymbol{\hat{\sigma}}^2\in\boldsymbol{\hat{\Sigma}}^2}{\max}v^{(\boldsymbol{\hat{\sigma}}^1,\boldsymbol{\hat{\sigma}}^2)}=\underset{\boldsymbol{\hat{\sigma}}^2\in\boldsymbol{\hat{\Sigma}}^2}{\max}\underset{\boldsymbol{\hat{\sigma}}^1\in\boldsymbol{\hat{\Sigma}}^1}{\min}v^{(\boldsymbol{\hat{\sigma}}^1,\boldsymbol{\hat{\sigma}}^2)}
\end{align}
\end{theorem}
Therefore, Theorem \ref{thrm:minimax-exist} confirms that the value $\hat{v}$ for the game $\cG$ exists and constitutes a saddle point equilibrium. Correspondingly, we denote by $\bm{\hat{Q}}$ the corresponding action value of the game $\cG$ associated with $\hat{v}$. The Theorem is proved in several steps, namely by proving the convergence of a dynamic programming procedure to the solution of the game (Theorem \ref{theorem:joint-sol}) and then proving the optimality of the solutions (Proposition \ref{prop:mpe_result}).
 An important consequence of Theorem \ref{thrm:minimax-exist} is that at the stable point, each player best responds to the influence of the other, formally:
\begin{definition}
A joint strategy $(\sigma^1,\sigma^2)=\boldsymbol{\sigma}\in\boldsymbol{\Sigma}$ is a Markov perfect equilibrium strategy if no player can improve the expected return by changing their current policy. The solution is a stable fixed point in which each player optimally responds to the policies of other players in the system: Formally, $\forall i \in \cN, \forall \sigma'^i \in \Sigma^i$, we have
$
    v(s|\boldsymbol{\sigma}) - v(s|(\sigma'^i,\sigma^{-i})) \ge 0.
$
\end{definition}


We can now state a key result. The following theorem proves the convergence of a value-iteration approach to the solution.

\begin{theorem}\label{theorem:joint-sol}
Let $v:\mathcal{S}\to\mathbb{R}$ then the sequence of Bellman operators acting on $v$ converges to the solution of the game, that is to say for any $s\in \cS$ the following holds: $
\underset{k\to\infty}{\lim}T^kv(s)=\boldsymbol{\hat{v}}(s)$,
\end{theorem}
\textcolor{black}
Theorem \ref{theorem:joint-sol} proves the solution to $\cG$ can be obtained by computing the limit of a dynamic programming procedure.  The proof of the Theorem is deferred to the Appendix.
%


    
%

\begin{proposition} \label{prop:mpe_result}
Let $\boldsymbol{\hat{\sigma}}\in \boldsymbol{\Sigma}$ be a strategy generated by the procedure outlined in Theorem \ref{theorem:joint-sol}, then $\boldsymbol{\hat{\sigma}}$ is a minimax Markov perfect equilibrium policy.    
\end{proposition}
\begin{proof}
The proof is achieved using similar arguments as those presented in Theorem 2 of \citep{shapley1953stochastic} with some modifications.
Denote by the \textit{finite} game $\cG^k$ of $k<\infty$ steps in which Player $1$ (Player $2$) maximises (minimises) the following objective $
v_k^{\sigma^1,\sigma^2}(s)=\mathbb{E}\left[\sum_{t=0}^k \gamma^t\left\{\cR(s_t, a_t,b_t)- \cC(s_t, a_t,b_t)\right\}\Big|s_0=s\right]$. 
Suppose in the game $\hat{\cG}^k$, Player 1 is given a payoff of $\cR(s,\bma)+\cP^{\bma}_{s's}v^{\sigma^1,\sigma^2}(s')$ given $\bma\equiv(a,b)\sim (\sigma^1,\sigma^2)\in \Sigma^1\times\Sigma^2$ and for any given $s\in\cS$. Now the Markov strategy $\sigma^1(a|s)\in\Sigma^1$ guarantees Player 1 a payoff of $v_k^{\sigma^1,\sigma^2}(s)$. Now in the game $\cG^k$, after $n<\infty$ steps and using the strategy $\sigma^1(a|s)$ gives Player 1 an expected payoff of at least $v_k^{\sigma^1,\sigma^2}(s)-\gamma^{n-1}\max\limits_{a\in\cA,b\in\cB}\cP^{(a,b)}_{s's}v_{k-1}^{\sigma^1,\sigma^2}(s')\leq v_k^{\sigma^1,\sigma^2}(s)-\gamma^{n-1}\max\limits_{s'\in\cS}v_{k-1}^{\sigma^1,\sigma^2}(s') $. Therefore, accounting for the $n$ steps, the total payoff for Player 1 is at least $v_k^{\sigma^1,\sigma^2}(s)-\gamma^{n-1}\max\limits_{s'\in\cS}v_{k-1}^{\sigma^1,\sigma^2}(s')-\sum_{t=0}^{n-1}\gamma^tv_{n-t}^{\sigma^1,\sigma^2}(s')= v_k^{\sigma^1,\sigma^2}(s)-\gamma^{n-1}\max\limits_{s'\in\cS}v_{k-1}^{\sigma^1,\sigma^2}(s')-\gamma^n\frac{1-\gamma^n}{1-\gamma}\|v\|:=\bm{\tilde{v}}^{\sigma^1,\sigma^2}_{k,n}(s)$. This expression holds for arbitrarily large values of $n$ in particular $\lim\limits_{n\to\infty}\bm{\tilde{v}}^{\sigma^1,\sigma^2}_{k,n}(s)=v_k^{\sigma^1,\sigma^2}(s)$ from which it follows that the strategy $\sigma^1$ is optimal for Player 1. After using analogous arguments for Player 2 we deduce the result.
\end{proof}

Having constructed a procedure to find each player's best-response strategy, we now seek to determine the conditions when an intervention should be performed. Let us denote by $\{\tau_k\}_{k\geq 0}$ ($\{\rho_r\}_{r\geq 0}$) the points at which each Player 1 (Player 2) decides to act or \textit{intervention times}, so for example if Player 1 chooses to perform an action at state $s_6$ and again at state $s_8$, then $\tau_1=6$ and $\tau_2=8$. We say that the times $\{\tau_k\}_{k\geq 0}$ and $\{\rho_r\}_{r\geq 0}$ are \textit{best-response intervention times} if executing actions at that sequence of times supports an MPE strategy. The following result  characterises the best-response intervention times $\{\tau_k\}_{k\geq 0}$ and $\{\rho_r\}_{r\geq 0}$.

\begin{proposition}\label{prop:switching_times}
The Player 1 and Player 2 best-response intervention times
are given by the following
\begin{align*}\tau_k&=\inf\{\tau>\tau_{k-1}|\bm\cM_1\bm{\hat{Q}}= \bm{\hat{Q}}\},
\\\rho_r&=\inf\{\rho>\rho_{r-1}|\bm\cM_2\bm{\hat{Q}}= \bm{\hat{Q}}\}.
\end{align*}
\end{proposition}
 Therefore, given the function $\bm{\hat{Q}}$, the times $\{\tau_k\}, \{\rho_r\}$ can be determined by evaluating if $\bm\cM_i\bm{\hat{Q}}=\bm{\hat{Q}}$ hold.
 
 A key aspect of Prop. \ref{prop:switching_times} is that it exploits the cost structure of the problem to determine when each player should perform an intervention. In particular, the equality $\bm\cM_i\bm{\hat{Q}}=\bm{\hat{Q}}$ implies that performing an action and incurring a cost for doing so is optimal.  The following result characterises the action conditions:
\begin{corollary}\label{prop:switching_cond}
For any $s\in\mathcal{S}$, Player 1 performs an action whenever the following condition is satisfied: $
\boldsymbol{1}_{\mathbb{R}_+}\left(\boldsymbol{\cM}_1\hat{\bm{Q}}(s,\boldsymbol{a}|\cdot)-\max\limits_{a\in \cA}\hat{\bm{Q}}(s, a,b)\right)=1,\;\forall b\in\cB$,
where $\boldsymbol{1}_{\mathbb{R}_{+}}$ is the indicator function i.e. $\boldsymbol{1}_{\mathbb{R}_{+}}(x)=1$ if $x>0$ and $\boldsymbol{1}_{\mathbb{R}_{+}}(x)=0$ otherwise. Analogously, Player 2 performs an action whenever $
\boldsymbol{1}_{\mathbb{R}_+}\left(\min\limits_{a\in \cA}\hat{\bm{Q}}(s, a,b)-\boldsymbol{\cM}_2\hat{\bm{Q}}(s,\boldsymbol{a}|\cdot)\right)=1,\;\forall a\in\cA$.
\end{corollary}
 The result provides characterisation of where each player should execute an action. 


\section{Learning in Unknown Environments}\label{sec:learning}
In this section, we study the problem while considering settings in which the game $\cG$ is entirely unknown to the players. To devise a method that enables us to handle this setting, we employ tools from reinforcement learning through which each player learns its minimax equilibrium strategy by repeatedly playing the game.  Within the standard RL paradigm, RL agents decide on an action at each time step thus the standard RL framework is not designed to handle the current setting \citep{sutton2018reinforcement}.  In what follows, we introduce an RL method that accommodates impulse controls in the stochastic game setting enabling the players to learn their equilibrium strategies entirely through repeated interaction. This, in turn, enables the players to learn when to perform actions in addition to the best-response actions which together constitute a best-response strategy to the play of their opponent.

\begin{algorithm}[!ht]
\begin{algorithmic}[1] 
		\STATE {\bfseries Input:} Constant $\epsilon\geq 0$, 
		\STATE {\bfseries Initialise:} Q-function, $\bmQ_0$ 
		\REPEAT
  \STATE{$n\gets 0$}
		    \FOR{$t=0,1,\ldots$}
    		    \STATE Compute $a_{t}\in\arg\max \bmQ_n(s_{t},0,b_t), b_{t}\in\arg\min \bmQ_n(s_{t},a_t,0)$ 
    		    \IF{$\bm\cM_1\bmQ_n\geq\bmQ_n$} 
         \STATE Apply $a_{t}$ so $s_{t+1}\sim P(\cdot|a_t,0,s_t),$
                    
        		    \STATE Receive rewards $r_{t} = \cR(s_{t},a_{t},0)$
    		    \ELSE
          \IF{$\bm\cM_2\bmQ_n\leq\bmQ_n$}
    		     \STATE Apply $b_{t}$ so $s_{t+1}\sim P(\cdot|0,b_t,s_t),$
                `\STATE Receive rewards $r_{t} = \cR(s_{t},0,b_t)$
                    \ELSE
                    		     \STATE Apply no action so $s_{t+1}\sim P(\cdot|\bm0,s_t),$
        	\STATE Receive rewards $r_t = \cR(s_{t},\bm0)$.
    		    \ENDIF
                    \ENDIF
        	\ENDFOR
    	\STATE{\textbf{// Learn $\hat{\bmQ}$}}
        \STATE Update $\bmQ_n$ function according to the update rule \eqref{q_learning_update} 
         \UNTIL{$|\bmQ_{n}(s,\bma)-\bmQ_{n-1}(s,\bma)|\leq\epsilon, \forall s\in \mathcal{S}$}
	\caption{}
\label{algo:1} 
\end{algorithmic}         
\end{algorithm}

\label{sec:learning_analysis}
We initiate the study of learning the minimax solution to $\cG$ by presenting the main result of the section, namely Theorem \ref{theorem:q_learning}. The theorem proves that a variant of a Q-learning \citep{bertsekas2012approximate} procedure converges almost surely to the action-value $\bm\hat{Q}$ for the game $\cG$.

\begin{theorem}\label{theorem:q_learning}
Consider the following Q learning variant:
\begin{align}\nonumber
    \bmQ_{t+1}&(s_t,\boldsymbol{a})=\bmQ_t(s_t,\boldsymbol{a})
\\&\begin{aligned}+\boldsymbol{\alpha}_t(s_t,\boldsymbol{a}_t)\Big[\min\Bigg(\max\left\{\bm\cM_1\bmQ_t(s_t,\bma), \cR(s_t,\boldsymbol{0})+\gamma\bmQ_t(s_{t+1},\bm0)\right\}&\\,\bm\cM_2\bmQ_t(s_t,\bma)\Bigg)-\bmQ_t(s_t,\boldsymbol{a})\Big]&,
\end{aligned}`\label{q_learning_update}
\end{align}
then $\bmQ_t$ converges to $\bm{\hat{Q}}$ with probability $1$, where $s_t,s_{t+1}\in\cS$ and $\bma\in\cA\times\cB$.
\end{theorem}

The proof of Theorem \ref{theorem:q_learning} requires several intermediate results using stochastic approximation theory \citep{bertsekas2012approximate}. We first make use of the following result:
\begin{theorem}[Theorem 1, pg 4 in \citep{jaakkola1994convergence}]\label{thrm:jaakkola}
Let $\Xi_t(s)$ be a random process that takes values in $\mathbb{R}^n$ and given by the following:
\begin{align}
    \Xi_{t+1}(s)=\left(1-\alpha_t(s)\right)\Xi_{t}(s)\alpha_t(s)L_t(s),
\end{align}
then $\Xi_t(s)$ converges to $0$ with probability $1$ under the following conditions:
\begin{itemize}
\item[i)] $0\leq \alpha_t\leq 1, \sum_t\alpha_t=\infty$ and $\sum_t\alpha^2_t<\infty$
\item[ii)] $\|\mathbb{E}[L_t|\mathcal{F}_t]\|\leq \gamma \|\Xi_t\|$, with $\gamma <1$;
\item[iii)] ${\rm Var}\left[L_t|\mathcal{F}_t\right]\leq c(1+\|\Xi_t\|^2)$ for some $c>0$.
\end{itemize}
\end{theorem}
Therefore, Theorem \ref{thrm:jaakkola} serves as an important stepping stone to prove the convergence of the procedure in \eqref{q_learning_update} - in particular, we must show that for our construction of the Bellman operator and the procedure outline in \eqref{q_learning_update}, the conditions oof Theorem \ref{thrm:jaakkola} hold. 
To prove the result, we show (i) - (iii) hold. Condition (i) holds by choice of learning rate. It therefore remains to prove (ii) - (iii). We first prove (ii). For this, we consider our variant of the Q-learning update rule:
\begin{align*}
\bmQ_{t+1}&(s_t,a_t,b_t)=\bmQ_t(s_t,a_t,b_t)
\\&\begin{aligned}+\boldsymbol{\alpha}_t(s_t,\boldsymbol{a}_t)\Big[\min\Bigg(\max\left\{\bm\cM_1\bmQ(s_t,a_t,b_t), \cR(s_t,\boldsymbol{0})+\gamma\bmQ_t(s_{t+1},\bm0)\right\}&\\,\bm\cM_2\bmQ(s_t,a_t,b_t)\Bigg)-\bmQ_t(s_t,a_t,b_t)\Big]&,
\end{aligned}
\end{align*}
After subtracting $\bm{\hat{Q}}(s_t,a_t,b_t)$ from both sides and some manipulation we obtain that:
\begin{align*}
&\Xi_{t+1}(s_t,a_t,b_t)
\\&=(1-\alpha_t(s_t,a_t,b_t))\Xi_{t}(s_t,a_t,b_t)
\\&\begin{aligned}+\alpha_t(s_t,a_t,b_t)\Big[\min\Bigg(\max\left\{\bm\cM_1\bmQ_t(s_t,a_t,b_t), \cR(s_t,\bm0)+\gamma \bmQ_t(s_{t+1},\bm0)\right\}&\\,\bm\cM_2\bmQ(s_t,a,b)\Bigg)-\bm{\hat{Q}}(s_t,a_t,b_t)\Big]&,
\end{aligned}
\end{align*}
where $\Xi_r(s_t,a_t,b_t):=\bmQ_r(s_t,a_t,b_t)-\bm{\hat{Q}}(s_t,a_t,b_t)$.

Let us now define by 
\begin{align*}
L_t(s_{\tau_k},a,b):=\min\left(\max\left\{\bm\cM_1\bmQ_t(s_t,a_t,b_t), \cR(s_t,\bm0)+\gamma \bmQ_t(s_{t+1},\bm0)\right\},\bm\cM_2\bmQ(s_t,a,b)\right)&\\-\bm{\hat{Q}}(s_t,a,b)&.
\end{align*}
Then
\begin{align}
\Xi_{t+1}(s_t,a_t,b_t)=(1-\alpha_t(s_t,a_t,b_t))\Xi_{t}(s_t,a_t,b_t)+\alpha_t(s_t,a_t,b_t)L_t(s_{\tau_k},a,b).   
\end{align}

We now observe that
\begin{align}\nonumber
&\mathbb{E}\left[L_t(s_{\tau_k},a,b)|\mathcal{F}_t\right]
\\&\nonumber\begin{aligned}=\sum_{s'\in\mathcal{S}}P(s';a,s_{\tau_k})\min\Big(\max\left\{\bm\cM_1\bmQ_t(s_t,a_t,b_t), \cR(s_t,\bm0)+\gamma \bmQ_t(s_{t+1},\bm0)\right\}&\\,\bm\cM_2\bmQ(s_t,a,b)\Big)
-\bm{\hat{Q}}(s_{\tau_k},a)&
\end{aligned}
\\&= T \bmQ_t(s,\bma)-\bm{\hat{Q}}(s,\bma). \label{expectation_L}
\end{align}
Now, using the fixed point property that implies $\bm{\hat{Q}}=T \bm{\hat{Q}}$, we find that
\begin{align}\nonumber
    \mathbb{E}\left[L_t(s_{\tau_k},a,b)|\mathcal{F}_t\right]&=T \bmQ_t(s,\bma)-T \bm{\hat{Q}}(s,\bma)
    \\&\leq\left\|T \bmQ_t-T \bm{\hat{Q}}\right\|\nonumber
    \\&\leq \gamma\left\| \bmQ_t- \bm{\hat{Q}}\right\|_\infty=\gamma\left\|\Xi_t\right\|_\infty.
\end{align}
using the contraction property of $T$ established in Proposition \ref{lemma:bellman_contraction}. This proves (ii).

We now prove iii), that is
\begin{align}
    {\rm Var}\left[L_t|\mathcal{F}_t\right]\leq c(1+\|\Xi_t\|^2).
\end{align}
Now by \eqref{expectation_L} we have that
\begin{align*}
  &{\rm Var}\left[L_t|\mathcal{F}_t\right]
  \\&\begin{aligned}= {\rm Var}\Big[\min\left(\max\left\{\bm\cM_1\bmQ_t(s_t,a_t,b_t), \cR(s_t,\bm0)+\gamma \bmQ_t(s_{t+1},\bm0)\right\},\bm\cM_2\bmQ(s_t,a_t,b_t)\right)&
  \\-\bm{\hat{Q}}(s_t,\bma)\Big]&
  \end{aligned}
  \\&\begin{aligned}= \mathbb{E}\Bigg[\Bigg(\min\left(\max\left\{\bm\cM_1\bmQ_t(s_t,a_t,b_t), \cR(s_t,\bm0)+\gamma \bmQ_t(s_{t+1},\bm0)\right\},\bm\cM_2\bmQ(s_t,a_t,b_t)\right)&  \\-\bm{\hat{Q}}(s_t,\bma)-\big(T \bmQ_t(s,\bma)-\bm{\hat{Q}}(s,\bma)\big)\Bigg)^2\Bigg]&
  \end{aligned}
      \\&\begin{aligned}= \mathbb{E}\Big[\Big(\min\left(\max\left\{\bm\cM_1\bmQ_t(s_t,a_t,b_t), \cR(s_t,\bm0)+\gamma \bmQ_t(s_{t+1},\bm0)\right\},\bm\cM_2\bmQ(s_t,a_t,b_t)\right)&\\-T \bmQ_t(s,\bma)\Big)^2\Big]&
      \end{aligned}
    \\&= {\rm Var}\left[\min\left(\max\left\{\bm\cM_1\bmQ_t(s_t,a_t,b_t), \cR(s_t,\bm0)+\gamma \bmQ_t(s_{t+1},\bm0)\right\},\bm\cM_2\bmQ(s_t,a_t,b_t)\right)\right]
    \\&\leq c(1+\|\Xi_t\|^2),
\end{align*}
for some $c>0$ where the last line follows due to the boundedness of $Q$ (which follows from Assumptions 2 and 4). This concludes the proof of the theorem.
\subsection{Convergence using Linear Function Approximators}
In reinforcement learning, an important consideration is the use of function approximators for the functions being learned during the learning process. Function approximators enable parameterisation of these key functions after which the parameters can be updated according to an RL update rule \citep{sutton1999policy}. Neural networks are an important case of function approximators. In what follows, we study our learning framework with linear function approximators. Linear function approximators do not offer as powerful approximation capabilities since the family of neural network function approximators is dense in the function space (the universal approximation theorem \citep{lu2020universal,xu2014reinforcement}. However, despite their deficiencies, linear function approximators are an important class due to their simplicity and convexity properties that do not suffer from issues such as convergence to suboptimal stationary points \citep{jin2020provably,busoniu2017reinforcement}. Moreover, the following analysis serves as an important step for proving analogous results with other function approximator classes.
\begin{definition}
For any test function $\psi\in L_2$ , a projection operator $\Pi$ acting $\psi$ is defined by the following 
\begin{align*}
\Pi \psi:=\underset{\bar{\psi}\in\{\Phi r|r\in\mathbb{R}^p\}}{\arg\min}\left\|\bar{\psi}-\psi\right\|. 
\end{align*}    
\end{definition}
\begin{theorem}\label{primal_convergence_theorem}
Algorithm 1 converges to the stable point of  $\mathcal{G}$, 
moreover, given a set of linearly independent basis functions $\Phi=\{\phi_1,\ldots,\phi_p\}$ with $\phi_k\in L_2,\forall k$. Algorithm 1 converges to a limit point $\hat{r}\in\mathbb{R}^p$ which is the unique solution to  $\Pi \mathfrak{F} (\Phi \hat{r})=\Phi \hat{r}$ where
    $\mathfrak{F}\Lambda:=\cR+\gamma P \max\left[\min\{\bm\cM_1\Lambda,\Lambda\},\bm\cM_2\Lambda\right]$. Moreover, $\hat{r}$ satisfies the following:
    \begin{align}
    \left\|\Phi \hat{r} - \bm{\hat{Q}}\right\|\leq (1-\gamma^2)^{-1/2}\left\|\Pi \bm{\hat{Q}}-\bm{\hat{Q}}\right\|.
    \end{align}
\end{theorem}
The theorem establishes the convergence of Algorithm 1 to minimax equilibrium of $\mathcal{G}$ with the use of linear function approximators. The second statement bounds the proximity of the convergence point by the smallest approximation error that can be achieved given the choice of basis functions.

The theorem is proven using a set of results that we now establish. 
 First, we prove the following bound holds:

\begin{lemma}
For any $\bmQ\in L_2$ we have that
\begin{align}
    \left\|\mathfrak{F}\bmQ-\mathfrak{F}\bmQ'\right\|\leq \gamma\left\|\bmQ-\bmQ'\right\|,
\end{align}
so that the operator $\mathfrak{F}$ is a contraction.
\end{lemma}
\begin{proof}
Now, we first note that by result iv) in the proof of Proposition \ref{lemma:bellman_contraction}, we deduced that for any $\bmQ, v\in L_2$ we have that
\begin{align*}
    \left\|\bm\cM_i\bmQ-\left[ \cR(\cdot,\boldsymbol{0})+\gamma\mathcal{P}^{\boldsymbol{0}}v'\right]\right\|\leq \gamma\left\|v-v'\right\|,\quad \forall i \in\{1,2\}.
\end{align*} 

A trivial modification reveals that we can deduce that for any $\bmQ,\bm{\hat{Q}}\in L_2$:
    $\left\|\bm\cM_i\bmQ-\bm{\hat{Q}}\right\|\leq 
 \gamma\left\|\bmQ-\bm{\hat{Q}}\right\|$. Using the contraction property of $\bm\cM_i$ and results i)-iv) of Proposition \ref{lemma:bellman_contraction}, we readily deduce the following bound:
\begin{align}\max\left\{\left\|\bm\cM_i\bmQ-\bm{\hat{Q}}\right\|,\left\|\bm\cM_i\bmQ-\bm\cM_j\bm{\hat{Q}}\right\|\right\}\leq \gamma\left\|\bmQ-\bm{\hat{Q}}\right\|,\quad \forall i,j \in\{1,2\}.
\label{m_bound_q_twice}
\end{align}
    
We now observe that $\mathfrak{F}$ is a contraction. Indeed, since for any $\bmQ,\bmQ'\in L_2$ we have that:
\begin{align*}
&\left\|\mathfrak{F}\bmQ-\mathfrak{F}\bmQ'\right\|
\\&=\left\|\cR+\gamma P \max\left[\min\{\bm\cM_1\bmQ,\bmQ\},\bm\cM_2\bmQ\right]-\left(\cR+\gamma P \max\left[\min\{\bm\cM_1\bmQ',\bmQ'\},\bm\cM_2\bmQ'\right]\right)\right\|
\\&=\gamma\left\| P \max\left[\min\{\bm\cM_1\bmQ,\bmQ\},\bm\cM_2\bmQ\right]-P \max\left[\min\{\bm\cM_1\bmQ',\bmQ'\},\bm\cM_2\bmQ'\right]\right\|
\\&\leq\gamma\left\| P\right\|\left\| \max\left[\min\{\bm\cM_1\bmQ,\bmQ\},\bm\cM_2\bmQ\right]-\max\left[\min\{\bm\cM_1\bmQ',\bmQ'\},\bm\cM_2\bmQ'\right]\right\|
\\&\begin{aligned}\leq\gamma\max\Bigg\{\left\|\min\{\bm\cM_1\bmQ,\bmQ\}-\min\{\bm\cM_1\bmQ',\bmQ'\}\right\|,\left\|\bm\cM_2\bmQ-\min\{\bm\cM_1\bmQ',\bmQ'\}\right\|&
\\,\left\|\bm\cM_2\bmQ'-\min\{\bm\cM_1\bmQ,\bmQ\}\right\|,\left\|\bm\cM_2\bmQ'-\bm\cM_1\bmQ\right\|\Bigg\}&
\end{aligned}
\\&\leq\gamma \max\left\{\left\|\bmQ-\bmQ'\right\|,\gamma\left\|\bmQ-\bmQ'\right\|\right\}
\\&=\gamma \left\|\bmQ-\bmQ'\right\|
\end{align*}
using the Cauchy-Schwarz inequality, \eqref{m_bound_q_twice} and again using the non-expansiveness of $P$.
\end{proof}

We next show that the following two bounds hold:
\begin{lemma}\label{projection_F_contraction_lemma}
For any $\bmQ\in\mathcal{V}$ we have that
\begin{itemize}
    \item[i)] 
$\qquad\qquad
    \left\|\Pi \mathfrak{F}\bmQ-\Pi \mathfrak{F}\bar{\bmQ}\right\|\leq \gamma\left\|\bmQ-\bar{\bmQ}\right\|$,
    \item[ii)]$\qquad\qquad\left\|\Phi \hat{r} - \bm{\hat{Q}}\right\|\leq \frac{1}{\sqrt{1-\gamma^2}}\left\|\Pi \bm{\hat{Q}} - \bm{\hat{Q}}\right\|$. 
\end{itemize}
\end{lemma}
\begin{proof}
The first result is straightforward since as $\Pi$ is a projection it is non-expansive and hence:
\begin{align*}
    \left\|\Pi \mathfrak{F}\bmQ-\Pi \mathfrak{F}\bar{\bmQ}\right\|\leq \left\| \mathfrak{F}\bmQ-\mathfrak{F}\bar{\bmQ}\right\|\leq \gamma \left\|\bmQ-\bar{\bmQ}\right\|,
\end{align*}
using the contraction property of $\mathfrak{F}$. This proves i). For ii), we note that by the orthogonality property of projections we have that $\left\langle\Phi \hat{r} - \Pi \bm{\hat{Q}},\Phi \hat{r} - \Pi \bm{\hat{Q}}\right\rangle=0$, hence we observe that:
\begin{align*}
    \left\|\Phi \hat{r} - \bm{\hat{Q}}\right\|^2&\leq \left\|\Phi \hat{r} - \Pi \bm{\hat{Q}}\right\|^2+\left\|\hat{\bmQ} - \Pi \bm{\hat{Q}}\right\|^2
\\&=\left\|\Pi \mathfrak{F}\Phi \hat{r} - \Pi \bm{\hat{Q}}\right\|^2+\left\|\hat{\bmQ} - \Pi \bm{\hat{Q}}\right\|^2
\\&\leq\left\|\mathfrak{F}\Phi \hat{r} -  \bm{\hat{Q}}\right\|^2+\left\|\hat{\bmQ} - \Pi \bm{\hat{Q}}\right\|^2
\\&=\left\|\mathfrak{F}\Phi \hat{r} -  \mathfrak{F}\bm{\hat{Q}}\right\|^2+\left\|\hat{\bmQ}  - \Pi \bm{\hat{Q}}\right\|^2
\\&\leq\gamma^2\left\|\Phi \hat{r} -  \bm{\hat{Q}}\right\|^2+\left\|\hat{\bmQ} - \Pi \bm{\hat{Q}}\right\|^2,
\end{align*}
after which we readily deduce the desired result.
\end{proof}

\begin{lemma}
Define  the operator $H$ by the following: 
\\\begin{center}
$
  H\bm{Q}(s,a,b)=  \begin{cases}
			\bm\cM_1Q_1(s,a), & \text{if $\bm\cM_1Q_1(s,a)>\Phi \hat{r}>-\bm\cM_2Q_2(s,b)$,}\\
   			\bm\cM_2Q_2(s,a), & \text{if $\bm\cM_2Q_2(s,a)>\Phi \hat{r}>-\bm\cM_1Q_1(s,b)$,}\\
            \bm{Q}(s,\bm0), & \text{otherwise},
		 \end{cases}$
   \end{center}
where  we define $\tilde{\mathfrak{F}}$ by: $
    \tilde{\mathfrak{F}}\bmQ:=\cR +\gamma PH\bmQ$.

For any $\bmQ,\bar{\bmQ}\in L_2$ we have that
\begin{align}
    \left\|\tilde{\mathfrak{F}}\bmQ-\tilde{\mathfrak{F}}\bar{\bmQ}\right\|\leq \gamma \left\|\bmQ-\bar{\bmQ}\right\|
\end{align}
and hence $\tilde{\mathfrak{F}}$ is a contraction mapping.
\end{lemma}
\begin{proof}
Using \eqref{m_bound_q_twice}, we now observe that
\begin{align*}
    &\;\;\;\;\left\|\tilde{\mathfrak{F}}\bmQ-\tilde{\mathfrak{F}}\bar{\bmQ}\right\|
    \\&=\left\|\cR+\gamma PH\bmQ -\left(\cR+\gamma PH\bar{\bmQ}\right)\right\|
\\&\leq \gamma\left\|H\bmQ - H\bar{\bmQ}\right\|
\\&\begin{aligned}
\leq \gamma\Big\|\max\Big\{\bm\cM_1\bmQ-\bm\cM_1\bar{\bmQ},\bm\cM_1\bmQ-\bar{\bmQ},\bm\cM_1\bar{\bmQ}-\bmQ,\bm\cM_2\bmQ-\bm\cM_2\bar{\bmQ},\bmQ-\bar{\bmQ}&
\\,\bm\cM_2\bmQ-\bar{\bmQ},\bm\cM_2\bar{\bmQ}-\bmQ\Big\}\Big\|&
\end{aligned}
\\&\begin{aligned}
\leq \gamma\max\Big\{\left\|\bm\cM_1\bmQ-\bm\cM_1\bar{\bmQ}\right\|,\left\|\bm\cM_1\bmQ-\bar{\bmQ}\right\|,\left\|\bm\cM_1\bar{\bmQ}-\bmQ\right\|,\left\|\bm\cM_2\bmQ-\bm\cM_2\bar{\bmQ}\right\|,\left\|\bmQ-\bar{\bmQ}\right\|&
\\,\left\|\bm\cM_2\bmQ-\bar{\bmQ}\right\|,\left\|\bm\cM_2\bar{\bmQ}-\bmQ\right\|\Big\}&
\end{aligned}
\\&\leq \gamma\max\left\{\gamma\left\|\bmQ-\bar{\bmQ}\right\|,\left\|\bmQ-\bar{\bmQ}\right\|\right\}
\\&=\gamma\left\|\bmQ-\bar{\bmQ}\right\|,
\end{align*}
again using \eqref{m_bound_q_twice} and the non-expansive property of $P$.
\end{proof}
\begin{lemma}
Define by $\tilde{\bmQ}:=\cR+\gamma Pv^{\boldsymbol{\tilde{\sigma}}}$ where
\begin{align}
    v^{\boldsymbol{\tilde{\sigma}}}(s):= \min\Bigg[\max\left\{\bm\cM_1\bmQ^{\bm\sigma^1,\bm\sigma^2}(s,\bma), \cR(s, \bm0)+\gamma\mathbb{E}_{s'\sim P}\left[ v^{\bm\sigma^1,\bm\sigma^2}(s')\right]\right\}
,\bm\cM_2\bmQ^{\bm\sigma^1,\bm\sigma^2}(s, \bma)\Bigg], \label{v_tilde_definition}
\end{align}
then $\tilde{\bmQ}$ is a fixed point of $\tilde{\mathfrak{F}}\tilde{\bmQ}$, that is $\tilde{\mathfrak{F}}\tilde{\bmQ}=\tilde{\bmQ}$. 
\end{lemma}
\begin{proof}
We begin by observing that
\begin{align*}
H\tilde{\bmQ}(s,\bma)&=H\left(\cR(s,\bma)+\gamma \cP^{\bma}_{ss'} v^{\boldsymbol{\tilde{\sigma}}}(s')\right)    
\\&= \begin{cases}
			\bm\cM_1Q_1(s,a), & \text{if $\bm\cM_1Q_1(s,a)>\Phi \hat{r}>-\bm\cM_2Q_2(s,b)$,}\\
   			\bm\cM_2Q_2(s,a), & \text{if $\bm\cM_2Q_2(s,a)>\Phi \hat{r}>-\bm\cM_1Q_1(s,b)$,}\\
            \bm{Q}(s,\bm0), & \text{otherwise},
		 \end{cases}
\\&= \begin{cases}
				\bm\cM_1Q_1(s,a), & \text{if $\bm\cM_1Q_1(s,a)>\Phi \hat{r}>-\bm\cM_2Q_2(s,b)$,}\\
   			\bm\cM_2Q_2(s,a), & \text{if $\bm\cM_2Q_2(s,a)>\Phi \hat{r}>-\bm\cM_1Q_1(s,b)$,}\\
            \cR(s,\bm0)+\gamma Pv^{\boldsymbol{\tilde{\sigma}}}, & \text{otherwise},
		 \end{cases}
\\&=v^{\boldsymbol{\tilde{\sigma}}}(s).
\end{align*}
Hence,
\begin{align}
    \tilde{\mathfrak{F}}\tilde{\bmQ}=\cR+\gamma PH\tilde{\bmQ}=\cR+\gamma Pv^{\boldsymbol{\tilde{\sigma}}}=\tilde{\bmQ}. 
\end{align}
which proves the result.
\end{proof}
\begin{lemma}\label{value_difference_Q_difference}
The following bound holds:
\begin{align}
    \mathbb{E}\left[v^{\boldsymbol{\hat{\sigma}}}(s_0)\right]-\mathbb{E}\left[v^{\boldsymbol{\tilde{\sigma}}}(s_0)\right]\leq 2\left[(1-\gamma)\sqrt{(1-\gamma^2)}\right]^{-1}\left\|\Pi \bm{\hat{Q}} -\bm{\hat{Q}}\right\|.
\label{F_tilde_fixed_point}\end{align}
\end{lemma}
\begin{proof}

By definitions of $v^{\boldsymbol{\hat{\sigma}}}$ and $v^{\boldsymbol{\tilde{\sigma}}}$ (c.f \eqref{v_tilde_definition}) and using Jensen's inequality and the stationarity property we have that,
\begin{align}\nonumber
    \mathbb{E}\left[v^{\boldsymbol{\hat{\sigma}}}(s_0)\right]-\mathbb{E}\left[v^{\boldsymbol{\tilde{\sigma}}}(s_0)\right]&=\mathbb{E}\left[Pv^{\boldsymbol{\hat{\sigma}}}(s_0)\right]-\mathbb{E}\left[Pv^{\boldsymbol{\tilde{\sigma}}}(s_0)\right]
    \\&\leq \left|\mathbb{E}\left[Pv^{\boldsymbol{\hat{\sigma}}}(s_0)\right]-\mathbb{E}\left[Pv^{\boldsymbol{\tilde{\sigma}}}(s_0)\right]\right|\nonumber
    \\&\leq \left\|Pv^{\boldsymbol{\hat{\sigma}}}-Pv^{\boldsymbol{\tilde{\sigma}}}\right\|. \label{v_approx_intermediate_bound_P}
\end{align}
Now recall that $\tilde{\bmQ}:=\cR+\gamma Pv^{\boldsymbol{\tilde{\sigma}}}$ and $\bm{\hat{Q}}:=\cR+\gamma Pv^{\boldsymbol{\sigma\star}}$,  using these expressions in \eqref{v_approx_intermediate_bound_P} we find that 
\begin{align*}
    \mathbb{E}\left[v^{\boldsymbol{\hat{\sigma}}}(s_0)\right]-\mathbb{E}\left[v^{\boldsymbol{\tilde{\sigma}}}(s_0)\right]&\leq \frac{1}{\gamma}\left\|\tilde{\bmQ}-\bm{\hat{Q}}\right\|. \label{v_approx_q_approx_bound}
\end{align*}
Moreover, by the triangle inequality and using the fact that $\mathfrak{F}(\Phi \hat{r})=\tilde{\mathfrak{F}}(\Phi \hat{r})$ and that $\mathfrak{F}\bm{\hat{Q}}=\bm{\hat{Q}}$ and $\mathfrak{F}\tilde{\bmQ}=\tilde{\bmQ}$ (c.f. \eqref{F_tilde_fixed_point}) we have that
\begin{align*}
\left\|\tilde{\bmQ}-\bm{\hat{Q}}\right\|&\leq \left\|\tilde{\bmQ}-\mathfrak{F}(\Phi \hat{r})\right\|+\left\|\bm{\hat{Q}}-\tilde{\mathfrak{F}}(\Phi \hat{r})\right\|    
\\&\leq \gamma\left\|\tilde{\bmQ}-\Phi \hat{r}\right\|+\gamma\left\|\bm{\hat{Q}}-\Phi \hat{r}\right\| 
\\&\leq 2\gamma\left\|\tilde{\bmQ}-\Phi \hat{r}\right\|+\gamma\left\|\bm{\hat{Q}}-\tilde{\bmQ}\right\|, 
\end{align*}
which gives the following bound:
\begin{align*}
\left\|\tilde{\bmQ}-\bm{\hat{Q}}\right\|&\leq 2\gamma\left(1-\gamma\right)^{-1}\left\|\tilde{\bmQ}-\Phi \hat{r}\right\|, 
\end{align*}
from which, using Lemma \ref{projection_F_contraction_lemma}, we deduce that $
    \left\|\tilde{\bmQ}-\bm{\hat{Q}}\right\|\leq 2\gamma\left[(1-\gamma)\sqrt{(1-\gamma^2)}\right]^{-1}\left\|\Pi \bm{\hat{Q}} - \bm{\hat{Q}}\right\|$,
after which by \eqref{v_approx_q_approx_bound}, we finally obtain
\begin{align*}
        \mathbb{E}\left[v^{\boldsymbol{\hat{\sigma}}}(s_0)\right]-\mathbb{E}\left[v^{\boldsymbol{\tilde{\sigma}}}(s_0)\right]\leq  2\left[(1-\gamma)\sqrt{(1-\gamma^2)}\right]^{-1}\left\|\Pi \bm{\hat{Q}} - \bm{\hat{Q}}\right\|,
\end{align*}
as required.
\end{proof}

Let us rewrite the update in the following way:
\begin{align*}
    r_{t+1}=r_t+\gamma_t\Xi(w_t,r_t),
\end{align*}
where the function $\Xi:\mathbb{R}^{2d}\times \mathbb{R}^p\to\mathbb{R}^p$ is given by:
\begin{align*}
\Xi(w,r):=\phi(s)\left(\cR(s,\cdot)+\gamma\max\left\{\min\{(\Phi r) (s'),\bm\cM_1(\Phi r) (s')\},\bm\cM_2(\Phi r) (s')\right\}-(\Phi r)(s)\right),
\end{align*}
for any $w\equiv (s,s')\in\mathcal{S}^2$  and for any $r\in\mathbb{R}^p$. Let us also define the function $\boldsymbol{\Xi}:\mathbb{R}^p\to\mathbb{R}^p$ by the following:
\begin{align*}
    \boldsymbol{\Xi}(r):=\mathbb{E}_{w_0\sim (\mathbb{P},\mathbb{P})}\left[\Xi(w_0,r)\right]; w_0:=(s_0,z_1).
\end{align*}
\begin{lemma}\label{iteratation_property_lemma}
The following statements hold for all $z\in \{0,1\}\times \mathcal{S}$:
\begin{itemize}
    \item[i)] $
(r-\hat{r})\boldsymbol{\Xi}_k(r)<0,\qquad \forall r\neq \hat{r},    
$
\item[ii)] $
\boldsymbol{\Xi}_k(\hat{r})=0$.
\end{itemize}
\end{lemma}
\begin{proof}
To prove the statement, we first note that each component of $\boldsymbol{\Xi}_k(r)$ admits a representation as an inner product, indeed: 
\begin{align*}
\boldsymbol{\Xi}_k(r)&=\mathbb{E}\left[\phi_k(s_0)(\cR(s_0,\bma_0)+\gamma\max\left\{\min\{(\Phi r) (s_1),\bm\cM_1(\Phi r) (s_1)\},\bm\cM_2(\Phi r) (s_1)\right\}-(\Phi r)(s_0)\right] 
\\&=\mathbb{E}\left[\phi_k(s_0)(\cR(s_0,\bma_0)+\gamma\mathbb{E}\left[\max\left\{\min\{(\Phi r) (s_1),\bm\cM_1(\Phi r) (s_1)\},\bm\cM_2(\Phi r) (s_1)\right\}|x_0\right]-(\Phi r)(s_0)\right]
\\&=\mathbb{E}\left[\phi_k(s_0)(\cR(s_0,\bma_0)+\gamma P\max\left\{\min\{(\Phi r) ,\bm\cM_1(\Phi r) \},\bm\cM_2(\Phi r) \right\}(s_0)-(\Phi r)(s_0)\right]
\\&=\left\langle\phi_k,\mathfrak{F}\Phi r-\Phi r\right\rangle,
\end{align*}
using the iterated law of expectations and the definitions of $P$ and $\mathfrak{F}$.

We now are in a position to prove i). Indeed, we now observe the following:
\begin{align*}
\left(r-\hat{r}\right)\boldsymbol{\Xi}_k(r)&=\sum_{l=1}\left(r(l)-\hat{r}(l)\right)\left\langle\phi_l,\mathfrak{F}\Phi r -\Phi r\right\rangle
\\&=\left\langle\Phi r -\Phi \hat{r}, \mathfrak{F}\Phi r -\Phi r\right\rangle
\\&=\left\langle\Phi r -\Phi \hat{r}, (\boldsymbol{1}-\Pi)\mathfrak{F}\Phi r+\Pi \mathfrak{F}\Phi r -\Phi r\right\rangle
\\&=\left\langle\Phi r -\Phi \hat{r}, \Pi \mathfrak{F}\Phi r -\Phi r\right\rangle,
\end{align*}
where in the last step we used the orthogonality of $(\boldsymbol{1}-\Pi)$. We now recall that $\Pi \mathfrak{F}\Phi \hat{r}=\Phi \hat{r}$ since $\Phi \hat{r}$ is a fixed point of $\Pi \mathfrak{F}$. Additionally, using Lemma \ref{projection_F_contraction_lemma} we observe that $\|\Pi \mathfrak{F}\Phi r -\Phi \hat{r}\| \leq \gamma \|\Phi r -\Phi \hat{r}\|$. With this, we now find that
\begin{align*}
&\left\langle\Phi r -\Phi \hat{r}, \Pi \mathfrak{F}\Phi r -\Phi r\right\rangle    
\\&=\left\langle\Phi r -\Phi \hat{r}, (\Pi \mathfrak{F}\Phi r -\Phi \hat{r})+ \Phi \hat{r} -\Phi r\right\rangle
\\&\leq\left\|\Phi r -\Phi \hat{r}\right\|\left\|\Pi \mathfrak{F}\Phi r -\Phi \hat{r}\right\|- \left\|\Phi \hat{r} -\Phi r\right\|^2
\\&\leq(\gamma -1)\left\|\Phi \hat{r} -\Phi r\right\|^2,
\end{align*}
which is negative since $\gamma<1$ which completes the proof of part i).

The proof of part ii) is straightforward since we readily observe that
\begin{align*}
    \boldsymbol{\Xi}_k(\hat{r})= \left\langle\phi_l, \mathfrak{F}\Phi \hat{r}-\Phi r\right\rangle= \left\langle\phi_l, \Pi \mathfrak{F}\Phi \hat{r}-\Phi r\right\rangle=0,
\end{align*}
as required and from which we deduce the result.
\end{proof}
To prove the theorem, we make use of a special case of the following result:

\begin{theorem}[Th. 17, p. 239 in \citep{benveniste2012adaptive}] \label{theorem:stoch.approx.}
Consider a stochastic process $r_t:\mathbb{R}\times\{\infty\}\times\Omega\to\mathbb{R}^k$ which takes an initial value $r_0$ and evolves according to the following:
\begin{align}
    r_{t+1}=r_t+\alpha \Xi(s_t,r_t),
\end{align}
for some function $s:\mathbb{R}^{2d}\times\mathbb{R}^k\to\mathbb{R}^k$ and where the following statements hold:
\begin{enumerate}
    \item $\{s_t|t=0,1,\ldots\}$ is a stationary, ergodic Markov process taking values in $\mathbb{R}^{2d}$
    \item For any positive scalar $q$, there exists a scalar $\mu_q$ such that $\mathbb{E}\left[1+\|s_t\|^q|s\equiv s_0\right]\leq \mu_q\left(1+\|s\|^q\right)$
    \item The step size sequence satisfies the Robbins-Monro conditions, that is $\sum_{t=0}^\infty\alpha_t=\infty$ and $\sum_{t=0}^\infty\alpha^2_t<\infty$
    \item There exists scalars $d$ and $q$ such that $    \|\Xi(w,r)\|
        \leq d\left(1+\|w\|^q\right)(1+\|r\|)$
    \item There exists scalars $d$ and $q$ such that $
        \sum_{t=0}^\infty\left\|\mathbb{E}\left[\Xi(w_t,r)|z_0\equiv z\right]-\mathbb{E}\left[\Xi(w_0,r)\right]\right\|
        \leq d\left(1+\|w\|^q\right)(1+\|r\|)$
    \item There exists a scalar $d>0$ such that $
        \left\|\mathbb{E}[\Xi(w_0,r)]-\mathbb{E}[\Xi(w_0,\bar{r})]\right\|\leq d\|r-\bar{r}\| $
    \item There exists scalars $d>0$ and $q>0$ such that $
        \sum_{t=0}^\infty\left\|\mathbb{E}\left[\Xi(w_t,r)|w_0\equiv w\right]-\mathbb{E}\left[\Xi(w_0,\bar{r})\right]\right\|
        \leq c\|r-\bar{r}\|\left(1+\|w\|^q\right) $
    \item There exists some $\hat{r}\in\mathbb{R}^k$ such that $\boldsymbol{\Xi}(r)(r-\hat{r})<0$ for all $r \neq \hat{r}$ and $\bar{s}(\hat{r})=0$. 
\end{enumerate}
Then $r_t$ converges to $\hat{r}$ almost surely.
\end{theorem}

In order to apply the Theorem \ref{theorem:stoch.approx.}, we show that conditions 1 - 7 are satisfied.

\begin{proof}
Conditions 1-2 are true by assumption while condition 3 can be made true by choice of the learning rates. Therefore it remains to verify conditions 4-7 are met.   

To prove 4, we observe that
\begin{align*}
\left\|\Xi(w,r)\right\|
&=\left\|\phi(s)\left(\cR(s,\cdot)+\gamma\max\left\{\min\{(\Phi r) (s'),\bm\cM_1(\Phi r) (s')\},\bm\cM_2(\Phi r) (s')\right\}-(\Phi r)(s)\right)\right\|
\\&\leq\left\|\phi(s)\right\|\left\|\cR(s,\cdot)+\gamma\left(\left\|\phi(s')\right\|\|r\|+\bm\cM_1\Phi (s')+\bm\cM_2\Phi (s')\right)\right\|+\left\|\phi(s)\right\|\|r\|
\\&\leq\left\|\phi(s)\right\|\left(\|\cR(s,\cdot)\|+\gamma\|\bm\cM_1\Phi (s')\|+\gamma\|\bm\cM_2\Phi (s')\|\right)+\left\|\phi(s)\right\|\left(\gamma\left\|\phi(s')\right\|+\left\|\phi(s)\right\|\right)\|r\|.
\end{align*}
%
%
Now using the definition of $\bm\cM_i$, we readily observe that $\|\bm\cM_i\Phi (s')\|\leq \|c\|_\infty+  \|\cR\|+\gamma\|\cP^{\bm\sigma}_{s's_t}\Phi\|\leq \|c\|_\infty+\| \cR\|+\gamma\|\Phi\|$ for $i=1,2$ using the non-expansiveness of $P$.

Hence, we lastly deduce that
\begin{align*}
\left\|\Xi(w,r)\right\|
&\leq\left\|\phi(s)\right\|\left(\|\cR(s,\cdot)\|+\gamma\|\bm\cM_1\Phi (s')\|+\gamma\|\bm\cM_2\Phi (s')\|\right)+\left\|\phi(s)\right\|\left(\gamma\left\|\phi(s')\right\|+\left\|\phi(s)\right\|\right)\|r\|
\\&\leq\left\|\phi(s)\right\|\left(\|\cR(s,\cdot)\|+2\gamma\left(\|c\|_\infty+\| \cR\|+|\phi\|\right)\right)+\left\|\phi(s)\right\|\left(\gamma\left\|\phi(s')\right\|+\left\|\phi(s)\right\|\right)\|r\|,
\end{align*}
we then easily deduce the result using the boundedness of $\phi$ and $\cR$.

Now we observe the following Lipschitz condition on $\Xi$:
\begin{align*}
&\left\|\Xi(w,r)-\Xi(w,\bar{r})\right\|
\\&\begin{aligned}=\Big\|\phi(s)\Big(\gamma\max\left\{\min\{(\Phi r) (s'),\bm\cM_1\Phi (s')\},\bm\cM_2\Phi  (s')\right\}-\gamma\max\big\{\min\{(\Phi \bar{r}) (s'),\bm\cM_1\Phi (s')\},&\\\bm\cM_2\Phi  (s')\big\}\Big)
-\left((\Phi r)(s)-\Phi\bar{r}(s)\right)\Big\|&
\end{aligned}
\\&\leq\gamma\left\|\phi(s)\right\|\Big\|\max\left\{\min\left\{\phi'(s') r,\bm\cM_1\Phi'(s')\right\},\bm\cM_2\Phi'(s')\right\}
\\&\qquad\qquad\qquad-\max\left\{\min\left\{(\phi'(s') \bar{r}),\bm\cM_1\Phi'(s')\right\},\bm\cM_2\Phi'(s')\right\}\Big\|+\left\|\phi(s)\right\|\left\|\phi'(s) r-\phi(s)\bar{r}\right\|
\\&\leq\gamma\left\|\phi(s)\right\|\left\|\min\left\{\phi'(s') r,\bm\cM_1\Phi'(s')\right\}-\min\left\{(\phi'(s') \bar{r}),\bm\cM_1\Phi'(s')\right\}\right\|+\left\|\phi(s)\right\|\left\|\phi'(s) r-\phi(s)\bar{r}\right\|
\\&\leq\gamma\left\|\phi(s)\right\|\left\|\phi'(s') r-\phi'(s') \bar{r}\right\|+\left\|\phi(s)\right\|\left\|\phi'(s) r-\phi'(s)\bar{r}\right\|
\\&\leq \left\|\phi(s)\right\|\left(\gamma\left\|\phi'(s')\right\|+ \left\|\phi'(s)\right\|\right)\left\|r-\bar{r}\right\|
\\&\leq {\rm const.}\left\|r-\bar{r}\right\|,
\end{align*}
using Cauchy-Schwarz inequality and  that for any scalars $a,b,c$ we have that \\$
    \left|\max\{a,b\}-\max\{b,c\}\right|\leq \left|a-c\right|$ and $
    \left|\min\{a,b\}-\min\{b,c\}\right|\leq \left|a-c\right|$, which proves Part 6.
    
Using Assumptions 3 and 4, we therefore deduce that
\begin{align}
\sum_{t=0}^\infty\left\|\mathbb{E}\left[\Xi(w,r)-\Xi(w,\bar{r})|w_0=w\right]-\mathbb{E}\left[\Xi(w_0,r)-\Xi(w_0,\bar{r})\right\|\right]\leq {\rm const.}\left\|r-\bar{r}\right\|(1+\left\|w\right\|^l).
\end{align}
which proves Part 7.

Part 2 is assured by Lemma \ref{projection_F_contraction_lemma} while Part 4 (and hence Part 5) is assured by Lemma \ref{value_difference_Q_difference} and lastly Part 8 is assured by Lemma \ref{iteratation_property_lemma}.
This result completes the proof of Theorem \ref{theorem:joint-sol}. 
\end{proof}

\section{Intervention Budget Constraints}
\label{sec:budget}



So far we have considered the case in which each player chooses a strategy that determines the conditions for 
the set of states in which the player will execute an action. The action cost function plays a critical role in influencing the willingness of the player to execute an action.   In this section, we consider an alternative setting in which the players face a budgetary constraint on the actions executed. A degenerate case of the budgetary formalism is a constraint on the number of actions the player can execute over the horizon.   We show that by tracking its remaining budget the Algorithm 2 framework is able to learn a policy that makes optimal usage of its budget while respecting the budget constraint almost surely.  

The program for the new SG is given by the following for any $s\in\mathcal{S}$:
\begin{align}
        v^{\bm\sigma^1,\bm\sigma^2}(s) \geq v^{\bm\sigma'^1,\bm\sigma^2}(s)\;\; 
        &\text{s.t. } n_1 - \sum^\infty_{t=0}\sum_{k\geq 1} \delta^t_{\tau_k}\geq 0, \quad \forall \bm\sigma'^1\in\bm\Sigma^1, \forall \bm\sigma^2\in\bm\Sigma^2 
        \\
v^{\bm\sigma^1,\bm\sigma^2}(s)\geq v^{\boldsymbol{\sigma}^1,\bm{\sigma'}^2}(s)\;\; 
        &\text{s.t. } n_2 - \sum^\infty_{t=0}\sum_{r\geq 1} \delta^t_{\rho_r}\geq 0, \quad \forall \bm\sigma^1\in\bm\Sigma^1, \forall \bm\sigma'^2\in\bm\Sigma^2 
        \label{capped_problem}
\end{align}
where for each $i\in\{1,2\}$, the quantity $n_i\in \mathbb{N}$ is a fixed value that represents the maximum number of allowed interventions \revision{and $\sum_{k\geq 1}\delta^t_{\tau_k}$ ($\sum_{r\geq 1}\delta^t_{\rho_r}$) is equal to one if a Player 1 (Player 2) action was applied at time $t$ and zero if it was not.}  

Hence, the value of the game for the new SG adheres to the following program $s\in\mathcal{S}$:
\begin{align}\nonumber
        v(s)=\max\limits_{\bm\sigma^1\in\bm\Sigma^1}\min\limits_{\bm\sigma^2\in\bm\Sigma^2}\;v^{\bm\sigma^1,\bm\sigma^2}(s)&=\min\limits_{\bm\sigma^2\in\bm\Sigma^2}\max\limits_{\bm\sigma^1\in\bm\Sigma^1}\; v^{\bm\sigma^1,\bm\sigma^2}(s)\;\; 
        \\&\text{s.t. } n_1 - \sum^\infty_{t=0}\sum_{k\geq 1} \delta^t_{\tau_k}\geq 0,  n_2 - \sum^\infty_{t=0}\sum_{r\geq 1} \delta^t_{\rho_r}\geq 0,
        \label{capped_problem_value}
\end{align}
To do this, we combine the above impulse control technology with state augmentation technique ~\citep{sootla2022saute}.  
As in \citep{sootla2022saute,mguni2022timing}, we introduce new variables $y_t$ and $z_t$ that track the remaining number of actions for Player 1 and Player 2 respectively so that: $y_t := n_1 - \sum_{t\geq 0}\sum_{k\geq 1} \delta^t_{\tau_k}$ and $z_t := n_2 - \sum_{t\geq 0}\sum_{r\geq 1} \delta^t_{\rho_r}$  where the variables
$y_t$ and $z_t$ are treated as new state variables which are components in an augmented state space $\boldsymbol{\cX}:=\cS\times\mathbb{N}^2$. We introduce the associated reward functions $\widetilde{\cR}:\boldsymbol{\cX}\times\boldsymbol{\cA}\to\cP(D)$ and the probability transition  function $\widetilde{P}:\boldsymbol{\cX}\times\boldsymbol{\cA}\times\boldsymbol{\cX}\to[0,1]$ whose state space input is now replaced by $\boldsymbol{\cX}$ for the game $\widetilde{\cG}=\langle \cN,\cS,\cA,\cB,\tilde{P},\tilde{\cR},\gamma\rangle$. We now prove Algorithm 2 ensures generates best-response minimax strategies within the constrained SG $\cG$.

\begin{align}\nonumber
        s_{t+1} \sim P(\cdot | s_t, a_t),\qquad
        y_{t+1} = y_t -& \sum_{k\geq 1}\delta^t_{\tau_k}, \quad y_0 = n_1 ,\qquad
        z_{t+1} = z_t - \sum_{k\geq 1}\delta^t_{\tau_k} , \quad z_0 = n_2.
        \\& \bmx_{t+1}\sim\widetilde\cP(\cdot|\bmx_t,\bma_t), \qquad \bmx_t\equiv(s_t,y_t,z_t)\in\bm\cX\label{impulse_saute_mdp}
\end{align}
where the output $\widetilde\cP$ is the sampled tuple $[P(\cdot | s_t, a_t),y_t - \sum_{k\geq 1}\delta^t_{\tau_k},z_t - \sum_{k\geq 1}\delta^t_{\tau_k}]$. 

To avoid violations, the reward function can be reshaped in the following manner $
    \widetilde \cR(s_t, y_t, z_t, a_t,b_t) = \begin{cases}
        \cR(s_t, a_t,b_t) & y_t,z_t \ge 0 ,\\
        -\infty & y_t <0,z_t>0,\\
                \infty & y_t >0,z_t<0,\\
        0 & y_t, z_t <0,
    \end{cases}
$

\noindent The value function is now given by the following expression:
\begin{align}\label{saute_impulse_problem}
\widetilde v^{\bm\sigma^1,\bm\sigma^2}(\bmx)  = \mathbb{E}\left[\sum_{t=0}^\infty \gamma^t \widetilde \cR(\bmx_t, a_t,b_t)|a_t\sim\bm{\widetilde{\sigma}}^1(\cdot|\bmx_t),b_t\sim\bm{\widetilde{\sigma}}^2(\cdot|\bmx_t), \bmx_t=(s_t,y_t,z_t)\right],
\end{align}
where the strategy for Player 1 and Player 2 now depends on the variable $y_t$ and $z_t$ respectively. Note that $\widetilde \cP$ in Equation~\ref{impulse_saute_mdp} is a Markov process and, the rewards $\widetilde \cR$ are bounded. Therefore, we can apply directly the results for impulse control to this case as well. We denote the augmented SG by $\widetilde \cG = \langle\boldsymbol{\cX}, \cA, \widetilde \cP, \widetilde R, \gamma\rangle$. We have the following. 
\begin{theorem} \label{thm:maxmin_budget} Consider 
the budgeted problem  $\widetilde \cG$, then: 

a) The Bellman equation holds, i.e. there exists a function $\tilde{v}$ such that for any $\bmx\in\bm\cX$
\begin{align}
\tilde{v}(\bmx) =
\min\Bigg[\max\left\{\bm\cM_1\widetilde \bmQ(\bmx,a,b), \widetilde \cR(\bmx, \bm0)+\gamma\mathbb{E}_{\bmx' \sim\widetilde\cP}\left[ \tilde{v}(\bmx')\right]\right\}
,\bm\cM_2\tilde{\bmQ}(\bmx,a,b)\Bigg]
\end{align}
where each player $i$ has a best-response policy $\widetilde \cG$ that takes the form $\widetilde{\boldsymbol{\sigma}}(\cdot | \bmx)$;

b) Given a function $\widetilde{v}:\cS\times \cZ\to \mathbb{R}$, the stable point solution for $\widetilde{\cG}$ is a given by 
\begin{align}
\underset{k\to\infty}{\lim}\tilde{T}^k\widetilde{\bmv}^{\tilde{\bm\sigma}^1,\tilde{\bm\sigma}^2}=\underset{{\bm{\tilde{\sigma}}}^1\in\tilde{\bm{\Sigma}}^1}{\max}\underset{\bm{\tilde{\sigma}}^2\in \tilde{\bm{\Sigma}}^2}{\min}\widetilde \bmv^{\bm{\hat{\sigma}}^1,\bm{\tilde{\sigma}}^2}=\underset{\bm{\tilde{\sigma}}^2\in \bm{\tilde{\Sigma}}^2}{\min}\underset{{\bm{\tilde{\sigma}}}^1\in\bm{\tilde{\Sigma}}^1}{\max}\widetilde \bmv^{\bm{\hat{\sigma}}^1,\bm{\hat{\sigma}}^2}=\hat{\widetilde \bmv},
\end{align}
where $(\bm{\hat{\tilde{\sigma}}}^1,\bm{\hat{\tilde{\sigma}}}^2)$ is an equilibrium $\widetilde\cG$ and $\tilde{T}$ is the Bellman operator of $\widetilde\cG$.
\end{theorem}

The result has several important implications. The first is that we can use a modified version of Algorithm 1 to obtain the solution of the problem while guaranteeing convergence (under standard assumptions). Secondly, the state augmentation procedure admits a Markovian representation of the best response strategies. 

Theorem \ref{thm:maxmin_budget}  shows Algorithm 2 converges to the minimax equilibrium when the players face an action-budget constraint.

\begin{algorithm}[!ht]
\begin{algorithmic}[1] 
		\STATE {\bfseries Input:} Constant $\epsilon\geq 0$, 
		\STATE {\bfseries Initialise:} Q-function, $\bmQ_0$ 
		\REPEAT
  \STATE{$n\gets 0$}
		    \FOR{$t=0,1,\ldots$}
    		    \STATE Compute $a_{t}\in\arg\max \widetilde\bmQ_n(\bmx_{t},0,b_t), b_{t}\in\arg\min \widetilde\bmQ_n(\bmx_{t},a_t,0)$ 
    		    \IF{$\bm\cM_1\widetilde\bmQ_n\geq\widetilde\bmQ_n$} 
         \STATE Apply $a_{t}$ so $\bmx_{t+1}\sim \widetilde 
         \cP(\cdot|\bmx_t,a_t,0),$
                    
        		    \STATE Receive rewards $r_{t} = \widetilde\cR(\bmx_t, a_{t},0)$
    		    \ELSE
          \IF{$\bm\cM_2\widetilde\bmQ_n\leq\widetilde\bmQ_n$}
    		     \STATE Apply $b_{t}$ so $\bmx_{t+1}\sim \widetilde \cP(\cdot|\bmx_t,0,b_t),$
                `\STATE Receive rewards $r_{t} = \widetilde\cR(s_{t},0,b_t)$
                    \ELSE
                    		     \STATE Apply no action so $\bmx_{t+1}\sim \widetilde \cP(\cdot|\bmx_t,\bm0),$
        	\STATE Receive rewards $r_t = \widetilde\cR(\bmx_{t},\bm0)$.
    		    \ENDIF
                    \ENDIF
        	\ENDFOR
    	\STATE{\textbf{// Learn $\hat{\bmQ}$}}
        \STATE Update $\widetilde\bmQ_n$ function according to the update rule \eqref{q_learning_update} 
         \UNTIL{$|\widetilde\bmQ_{n}(\bmx,\bma)-\widetilde\bmQ_{n-1}(\bmx,\bma)|\leq\epsilon$}
	\caption{}
\label{algo:1b} 
\end{algorithmic}         
\end{algorithm}

\section{Discussion on Subcases}
As discussed earlier, stochastic games with minimally bounded action costs studied in this paper have a number of applications within economics, financial systems and computer science. In addition to direct application, there are a number of important subcases that emerge as degenerate cases of our general framework. 

\subsubsection*{Games of Impulse Control and Stopping~\citep{mguni2018viscosity,campi2020nonzero}}
In this setting, two players engage in a zero sum game in which one player's actions are subject to minimally bounded costs (thus the player executes their actions using impulse controls) and the other player's only action is to decide when to stop the game. This induces a game of impulse control and stopping. This is a degenerate case of the stochastic game with minimally bounded costs when one of the players faces an action budget that caps the number of allowed actions to $1$.

A general case of this includes a flow payoff which is a payment received by the players after terminating the game.  In particular, the game $\cG$ degenerates into this case 

\subsubsection*{Dynkin games (Optimal Stopping Games) \citep{ekstrom2008optimal, martyr2016solving, martyr2015optimal}}
In this setting, two players engage in a zero sum game in which each player's only action is to decide when to stop the game. When this action is executed by a player, the game is arrested for both players at which point the players receive a terminal payoff thus the players must choose a time to stop which maximises their individual expected payoff contingent on the decision of the other player which gives rise to a strategic interaction. This is a degenerate case of the stochastic game with minimally bounded costs. A general case of this includes a flow payoff which is a payment received by the players after terminating the game.  In particular, the game $\cG$ degenerates into this case when the number of actions for each player is capped at $1$.

\noindent\textit{Optimal Stopping Problems.} An obvious yet well-known degenerate case of this subproblem are optimal stopping problems (in discrete-time). This case arises when one of the player's strategies is fixed in the optimal stopping game.

\subsubsection*{Impulse control (in discrete-time) \citep{bensoussan2008impulse}} An obvious yet well-known degenerate case of this subproblem are impulse Control problems (in discrete-time). This case arises when one of the player's strategies is fixed in the game $\cG$.

\section{Related Work}

In continuous-time optimal control theory \citep{oksendal2003stochastic,todorov2006optimal}, problems in which the player faces a cost for each action are tackled with a form of policy known as \textit{impulse control} \citep{belak2017general,mguni2018viscosity}. In impulse control frameworks, the dynamics of the system are modified through a sequence of discrete actions or bursts chosen at times that the player chooses to apply the control policy. This distinguishes impulse control models from classical decision methods in which a player takes actions at each time step while being tasked with the decision of only which action to take. Impulse control models represent appropriate modelling frameworks for financial environments with transaction costs, liquidity risks and economic environments in which players face fixed adjustment costs (e.g. \textit{menu costs}) \citep{midrigan2011menu}.  However, as yet the literature on learning environments with impulse control remains sparse \citep{mguni2022timing} and extension to the multi-player settings with learning is at present absent.


The current setting is intimately related to the \textit{optimal stopping problem} which widely occurs in finance, economics, and computer science \citep{mguni2019cutting,tsitsiklis1999optimal}. In the optimal stopping problem, the task is to determine a criterion that determines when to arrest the system and receive a terminal reward. In this case, standard (MA)RL methods are unsuitable since they require an expensive sweep (through the set of states) to determine the optimal point to arrest the system.   The impulse control problem can be viewed as an augmented problem of optimal stopping since the player must now determine both a sequence of points to perform an action or \textit{intervene} and their optimal magnitudes --- only acting when the cost of action is justified \citep{oksendal2019approximating}. Adapting RL to tackle optimal stopping problems has been widely studied \citep{tsitsiklis1999optimal,becker2018deep, chen2020zap} and applied to a variety of real-world settings within finance \citep{fathan2021deep}, radiation therapy \citep{ajdari2019towards} and network operating systems \citep{alcaraz2020online}. 

Closely related to impulse control, switching control structures have been studied in learning contexts to improve learning efficiency in reinforcement learning \citep{mguni2022ligs, mguni2023learning, pmlr-v202-mguni23a}. ROSA \citep{mguni2023learning} and LIGS \citep{mguni2022ligs} incorporate a dual switching method to activate a reward-shaping module to promote state visitations and coordination between adaptive agents in an RL and MARL respectively. LICRA \citep{mguni2022timing} adds a trainable switch to decide whether to use a costly execution-policy system to generate actions. Similarly, MANSA \citep{pmlr-v202-mguni23a} has an additional switch to decide whether to activate centralised training, a computationally expensive learning mode that facilitates coordination among adaptive agents. 

\section{Conclusion}
We presented a novel method to tackle the problem of learning how to select when to act in addition to learning which actions to execute.  Our framework, which is a general tool for tackling problems of this kind seamlessly adopts RL algorithms enabling them to efficiently tackle problems in which the player must be selective about when it executes actions. This is of fundamental importance in practical settings where performing many actions over the horizon can lead to costs and undermine the service life of machinery. We demonstrated that our solution, our framework which at its core has a sequential decision structure that first decides whether or not an action ought to be taken under the action policy can solve tasks where the player faces costs with extreme efficiency as compared to leading reinforcement learning methods. In some tasks, we showed that our framework is able to solve problems that are unsolvable using current reinforcement learning machinery. We envisage that this framework can serve as the basis for extensions to different settings including adversarial training for solving a variety of problems within RL.

\bibliography{main}
\bibliographystyle{iclr2023_conference}


\section{Appendix}\label{sec:proofs_appendix}

We begin the analysis with some preliminary lemmata and definitions that are useful for proving the main results.

\begin{definition}{A.1}
An operator $T: \mathcal{V}\to \mathcal{V}$ is said to be a \textbf{contraction} w.r.t a norm $\|\cdot\|$ if there exists a constant $c\in[0,1[$ such that for any $V_1,V_2\in  \mathcal{V}$ we have that:
\begin{align}
    \|TV_1-TV_2\|\leq c\|V_1-V_2\|.
\end{align}
\end{definition}

\begin{definition}{A.2}
An operator $T: \mathcal{V}\to  \mathcal{V}$ is \textbf{non-expansive} if $\forall v,V_2\in  \mathcal{V}$ we have:
\begin{align}
    \|Tv-TV_2\|\leq \|v-V_2\|.
\end{align}
\end{definition}

\begin{lemma} \label{max_lemma}
For any
$f: \mathcal{V}\to\mathbb{R},g: \mathcal{V}\to\mathbb{R}$, we have that:
\begin{align}
\left\|\underset{a\in \mathcal{V}}{\max}\:f(a)-\underset{a\in \mathcal{V}}{\max}\: g(a)\right\| \leq \underset{a\in \mathcal{V}}{\max}\: \left\|f(a)-g(a)\right\|.    \label{lemma_1_basic_max_ineq}
\end{align}
\end{lemma}
\begin{proof}
We restate the proof given in \citep{mguni2019cutting}:
\begin{align}
f(a)&\leq \left\|f(a)-g(a)\right\|+g(a)\label{max_inequality_proof_start}
\\\implies
\underset{a\in \mathcal{V}}{\max}f(a)&\leq \underset{a\in \mathcal{V}}{\max}\{\left\|f(a)-g(a)\right\|+g(a)\}
\leq \underset{a\in \mathcal{V}}{\max}\left\|f(a)-g(a)\right\|+\underset{a\in \mathcal{V}}{\max}\;g(a). \label{max_inequality}
\end{align}
Deducting $\underset{a\in \mathcal{V}}{\max}\;g(a)$ from both sides of (\ref{max_inequality}) yields:
\begin{align}
    \underset{a\in \mathcal{V}}{\max}f(a)-\underset{a\in \mathcal{V}}{\max}g(a)\leq \underset{a\in \mathcal{V}}{\max}\left\|f(a)-g(a)\right\|.\label{max_inequality_result_last}
\end{align}
After reversing the roles of $f$ and $g$ and redoing steps (\ref{max_inequality_proof_start}) - (\ref{max_inequality}), we deduce the desired result since the RHS of (\ref{max_inequality_result_last}) is unchanged.
\end{proof}
\begin{lemma}\label{max_l.val_inequality}
Define ${\rm val}^+[f]:=\min_{b\in\mathbb{B}}\max_{a\in\mathbb{A}}f(a,b)$ and define\\ ${\rm val}^-[f]:=\max_{a\in\mathbb{A}}\min_{b\in\mathbb{B}}f(a,b)$, then for any $b\in\mathbb{B}$ we have that for any $f,g\in \mathbb{L}$ and for any $c\in\mathbb{R}_{>0}$:
\begin{align}\nonumber
\left|\max_{a\in\mathbb{A}}f(a,b)-\max_{a\in\mathbb{A}}g(a,b)\right|&\leq c
\implies
\left|{\rm val}^-[f]-{\rm val}^-[g]\right|\leq c.
\end{align}
\end{lemma}
\begin{proof}[Proof of Lemma \ref{max_l.val_inequality}.]
We begin by noting the following inequality for any
$f:\mathcal{V}\times\mathcal{V}\to\mathbb{R},g:\mathcal{V}\times\mathcal{V}\to\mathbb{R}$ such that $f,g\in \mathbb{L}$ we have that for all $b\in\mathcal{V}$:
\begin{align}
\left|\underset{a\in\mathcal{V}}{\max}\:f(a,b)-\underset{a\in\mathcal{V}}{\max}\: g(a,b)\right| \leq \underset{a\in\mathcal{V}}{\max}\: \left|f(a,b)-g(a,b)\right|.    
\end{align}
From (\ref{lemma_1_basic_max_ineq}) we can straightforwardly derive the fact that for any $b\in\mathcal{V}$:
\begin{align}
\left|\underset{a\in\mathcal{V}}{\min}\: f(a,b)-\underset{a\in\mathcal{V}}{\min}\: g(a,b)\right| \leq \underset{a\in\mathcal{V}}{\max}\: \left|f(a,b)-g(a,b)\right|,   \label{lemma_1_max_ineq_min_version} 
\end{align}

Assume that for any $b\in\mathcal{V}$ the following inequality holds: 
\begin{align}
\underset{a\in\mathcal{V}}{\max}\: \left|f(a,b)-g(a,b)\right|  \leq c \label{lemma_1_assumption} 
\end{align}
Since (\ref{lemma_1_max_ineq_min_version}) holds for any $b\in\mathcal{V}$ and, by (\ref{lemma_1_basic_max_ineq}), we have in particular that
\begin{align}
&\nonumber\left|\underset{b\in\mathcal{V}}{\max}\;\underset{a\in\mathcal{V}}{\min}\: f(a,b)-\underset{b\in\mathcal{V}}{\max}\;\underset{a\in\mathcal{V}}{\min}\: g(a,b)\right| 
\\&\nonumber\leq
\underset{b\in\mathcal{V}}{\max}\left|\underset{a\in\mathcal{V}}{\min}\: f(a,b)-\underset{a\in\mathcal{V}}{\min}\: g(a,b)\right| 
\\&\leq \underset{b\in\mathcal{V}}{\max}\;\underset{a\in\mathcal{V}}{\max}\: \left|f(a,b)-g(a,b)\right|  \leq c,\label{lemma_1_max_ineq_fin} 
\end{align}
whenever (\ref{lemma_1_assumption}) holds which gives the required result.
\end{proof}

\begin{lemma}\label{max_min_inequality_2} For any $f,g,h\in \mathbb{L}$ and for any $c\in\mathbb{R}_{>0}$ we have that:
\begin{align}\nonumber
 \left\| f- g\right\|\leq c
\implies
\left\|\min\{ f,h\}-\min\{ g,h\}\right\|\leq c.
\end{align}
\end{lemma}
\begin{lemma}\label{max_triple_inequality} Let the functions $f,g,h\in \mathbb{L}$ then
\begin{align}
\left\|\max \{ f,h\}-\max\{g,h\}\right\|&\leq \|f-g\|.
\end{align}
\end{lemma}

\begin{lemma}{A.4}\label{non_expansive_P}
The probability transition kernel $P$ is non-expansive, that is:
\begin{align}
    \|Pv-PV_2\|\leq \|v-V_2\|.
\end{align}
\end{lemma} 
\begin{proof}
The result is well-known e.g. \citep{tsitsiklis1999optimal}. We give a proof using the Tonelli-Fubini theorem and the iterated law of expectations, we have that:
\begin{align*}
&\|PJ\|^2=\mathbb{E}\left[(PJ)^2[s_0]\right]=\mathbb{E}\left(\left[\mathbb{E}\left[J[s_1]|s_0\right]\right)^2\right]
\leq \mathbb{E}\left[\mathbb{E}\left[J^2[s_1]|s_0\right]\right] 
= \mathbb{E}\left[J^2[s_1]\right]=\|J\|^2,
\end{align*}
where we have used Jensen's inequality to generate the inequality. This completes the proof.
\end{proof}


The following result proves $T$ is a contraction:
\begin{proposition}\label{lemma:bellman_contraction}
The Bellman operator $T$ is a contraction, in particular, the following bound holds:
\begin{align}\nonumber
&\left\|Tv-Tv'\right\|\leq \gamma\left\|v-v'\right\|,
\end{align}
\end{proposition}
where $v,v'$ are elements of a finite normed vector space.

\section*{Proof of Theorem \ref{theorem:joint-sol}}

\begin{proof}[Proof of Proposition \ref{lemma:bellman_contraction}]
Recall, the Bellman operator $T$ acting on a function $v:\mathcal{S}\to\mathbb{R}$ is given by
\begin{align} 
T v(s):=\min\left[\max\Big\{\bm\cM_1 Q_1, \cR(s,\boldsymbol{0})+\gamma\sum_{s'\in\mathcal{S}}P(s';\boldsymbol{0},s)v(s')\Big\},\bm\cM_2Q_2\right].\label{bellman_proof_start} 
\end{align}
where $Q_1(s,a):=\bm{Q}(s,a,0)$ and $Q_2(s,b):=\bm{Q}(s,0,b)$.

In what follows and for the remainder of the script, we employ the following shorthands:
\begin{align*}
&\mathcal{P}^{\bma}_{s's}=:\sum_{s'\in\mathcal{S}}P(s';\bma,s), \quad\bm{\cP}^{\boldsymbol{\sigma}}_{ss'}=:\sum_{{\bma}\in\bm\cA}{\bm\sigma}(\bma|s)\mathcal{P}^{\bma}_{ss'}
\end{align*}

To prove that $T$ is a contraction, we first note that by \eqref{lemma_1_max_ineq_min_version}  we deduce that

\begin{align*}
    &\left|T v(s)-T v'(s)\right|:\\&=\Bigg|\min\left[\max\Big\{\bm\cM_1 Q_1, \cR(s,\boldsymbol{0})+\gamma\sum_{s'\in\mathcal{S}}P(s';\boldsymbol{0},s)v(s')\Big\},\bm\cM_2 Q_2\right]
    \\&\qquad\quad-\min\left[\max\Big\{\bm\cM_1 Q_1', \cR(s,\boldsymbol{0})+\gamma\sum_{s'\in\mathcal{S}}P(s';\boldsymbol{0},s)v'(s')\Big\},\bm\cM_2 Q_2'\right]\Bigg|
    \\&\leq \max\Bigg|\max\left\{\max\Big\{\bm\cM_1 Q_1, \cR(s,\boldsymbol{0})+\gamma\sum_{s'\in\mathcal{S}}P(s';\boldsymbol{0},s)v(s')\Big\},\bm\cM_2 Q_2\right\}
    \\&\qquad\quad-\max\left\{\max\Big\{\bm\cM_1 Q_1', \cR(s,\boldsymbol{0})+\gamma\sum_{s'\in\mathcal{S}}P(s';\boldsymbol{0},s)v'(s')\Big\},\bm\cM_2 Q_2'\right\}\Bigg|
\end{align*}

We now consider the four cases produced by \eqref{bellman_proof_start}, that is to say we prove the following statements:

i) $\qquad\qquad
\left| \cR(s_t,a_t,b_t)+\gamma\mathcal{P}^{\boldsymbol{0}}_{s's_t}v(s')-\left( \cR(s_t,a_t,b_t)+\gamma\mathcal{P}^{\boldsymbol{0}}_{s's_t}v'(s')\right)\right|\leq \gamma\left\|v-v'\right\|$

ii) $\qquad\qquad
\left\|\bm\cM_i\bmQ-\bm\cM_i\bmQ'\right\|\leq    \gamma\left\|v-v'\right\|,\qquad \qquad$.

iii) $\qquad\qquad
\left\|\bm\cM_i\bmQ-\bm\cM_j\bmQ'\right\|\leq    \gamma\left\|v-v'\right\|,\qquad \qquad$

iv)  $\qquad\qquad
    \left\|\bm\cM_i\bmQ-\left[ \cR(\cdot,\boldsymbol{0})+\gamma\mathcal{P}^{\boldsymbol{0}}_{s's_t}v'\right]\right\|\leq \gamma\left\|v-v'\right\|.
$

We begin by proving i).

Indeed, for any $\bma\in\bm\cA$ and $\forall s_t\in\mathcal{S}, \forall s'\in\mathcal{S}$ we have that 
\begin{align*}
&\left| \cR(s_t,\boldsymbol{0})+\gamma\mathcal{P}^{\boldsymbol{0}}_{s's_t}v(s')-\left[ \cR(s_t,\boldsymbol{0})+\gamma\mathcal{P}^{\boldsymbol{0}}_{s's_t}v'(s')\right]\right|
\\&\leq \underset{\boldsymbol{a}\in\boldsymbol{\mathcal{A}}}{\max}\;\left|\gamma\mathcal{P}^{\bma}_{s's_t}v(s')-\gamma\mathcal{P}^{\bma}_{s's_t}v'(s')\right|
\\&\leq \gamma\left\|Pv-Pv'\right\|
\\&\leq \gamma\left\|v-v'\right\|,
\end{align*}
again using the fact that $P$ is non-expansive and Lemma \ref{max_lemma}.

We now prove ii).

For any $\tau\in\mathcal{F}$, define by $\tau'=\inf\{t>\tau|s_t\in \cS_I,\tau\in\mathcal{F}_t\}$. We do the case where $i=1$ with the case $i=2$ being analogous. Now using the definition of $\bm\cM_1$ we have that for any $s_\tau\in\mathcal{S}$
\begin{align*}
&\left|(\bm\cM_1 \bmQ-\bm\cM_1 \bmQ')(s_\tau,a_\tau,b)\right|
\\&\leq \underset{a_\tau\in \mathcal{A}}{\max}    \Bigg|\cR(s_\tau,a_\tau,b)-c(s_\tau,a_\tau)+\gamma\cP_{s's_\tau}^{(a_\tau,b)}v(s')-\left(\cR(s_\tau,a_\tau,b)-c(s_\tau,a_\tau)+\gamma\cP_{s's_\tau}^{(a_\tau,b)}v'(s')\right)\Bigg| 
\\&\leq \max\limits_{a\in \mathcal{A},b\in\cB,s'\in\cS} \gamma\left|\cP_{s's_\tau}^{(a,b)}v(s')-\cP_{s's_\tau}^{(a,b)}v'(s')\right| 
\\&\leq \gamma\left\|Pv-Pv'\right\|
\\&\leq \gamma\left\|v-v'\right\|,
\end{align*}
using the fact that $P$ is non-expansive.

Next we prove (iii). We break the proof into two cases:

Case 1:
\begin{align}
\underset{a'\in \cA}{\max} \left(\cR(s_\tau,a',b)-c(s_\tau,a')+\gamma\cP^{(a',b)}_{s's_\tau}v(s')\right)-\underset{b'\in \cB}{\min} \left(\cR(s_\tau,a,b')+c(s_\tau,b')+\gamma\cP^{(a,b')}_{s's_\tau}v'(s')\right)\leq 0
\end{align}
\begin{align*}
&\left|(\bm\cM_1\bmQ-\bm\cM_2\bmQ')(s_{\tau},\bma)\right|
\\&=    \left|\underset{a'\in \cA}{\max} \left(\cR(s_\tau,a',b)-c(s_\tau,a')+\gamma\cP^{(a',b)}_{s's_\tau}v(s')\right)-\underset{b'\in \cB}{\min} \left(\cR(s_\tau,a,b')+c(s_\tau,b')+\gamma\cP^{(a,b')}_{s's_\tau}v'(s')\right)\right|
\\&\hspace{-1cm}\begin{aligned}\leq  \Bigg|\max\left\{\underset{a'\in \cA}{\max} \left(\cR(s_\tau,a',b)-c(s_\tau,a')+\gamma\cP^{(a',b)}_{s's_\tau}v(s')\right),\underset{b'\in \cB}{\min} \left(\cR(s_\tau,a,b')+c(s_\tau,b')+\gamma\cP^{(a,b')}_{s's_\tau}v(s')\right)\right\}&
\\-\underset{b'\in \cB}{\min}
\left(\cR(s_\tau,a,b')+c(s_\tau,b')+\gamma\cP^{(a,b')}_{s's_\tau}v'(s')\right)\Bigg|&
\end{aligned}
\\&\hspace{-1cm}\leq    \Bigg|\max\left\{\underset{a'\in \cA}{\max} \left(\cR(s_\tau,a',b)-c(s_\tau,a')+\gamma\cP^{(a',b)}_{s's_\tau}v(s')\right),\underset{b'\in \cB}{\min} \left(\cR(s_\tau,a,b')+c(s_\tau,b')+\gamma\cP^{(a,b')}_{s's_\tau}v(s')\right)\right\}
\\&\hspace{-1cm}- \max\left\{\underset{a'\in \cA}{\max} \left(\cR(s_\tau,a',b)-c(s_\tau,a')+\gamma\cP^{(a',b)}_{s's_\tau}v(s')\right),\underset{b'\in \cB}{\min} \left(\cR(s_\tau,a,b')+c(s_\tau,b')+\gamma\cP^{(a,b')}_{s's_\tau}v'(s')\right)\right\}
\\&\hspace{-1cm}\begin{aligned}+
\max\left\{\underset{a'\in \cA}{\max} \left(\cR(s_\tau,a',b)-c(s_\tau,a')+\gamma\cP^{(a',b)}_{s's_\tau}v(s')\right),\underset{b'\in \cB}{\min} \left(\cR(s_\tau,a,b')+c(s_\tau,b')+\gamma\cP^{(a,b')}_{s's_\tau}v'(s')\right)\right\}&
\\-\underset{b'\in \cB}{\min} \left(\cR(s_\tau,a,b')+c(s_\tau,b')+\gamma\cP^{(a,b')}_{s's_\tau}v'(s')\right)\Bigg|&
\end{aligned}
\\&\hspace{-1cm}\leq\left|    \underset{b'\in \cB}{\min} \left(\cR(s_\tau,a,b')+c(s_\tau,b')+\gamma\cP^{(a,b')}_{s's_\tau}v(s')\right)-\underset{b'\in \cB}{\min} \left(\cR(s_\tau,a,b')+c(s_\tau,b')+\gamma\cP^{(a,b')}_{s's_\tau}v'(s')\right)\right|
\\&\hspace{-1.5cm}+\Bigg|
\max\left\{\underset{a'\in \cA}{\max} \left(\cR(s_\tau,a',b)-c(s_\tau,a')+\gamma\cP^{(a',b)}_{s's_\tau}v(s')\right)-\underset{b'\in \cB}{\min} \left(\cR(s_\tau,a,b')+c(s_\tau,b')+\gamma\cP^{(a,b')}_{s's_\tau}v'(s')\right),0\right\}\Bigg|
\\&   
\leq\gamma\max\limits_{a'\in \cA}\max\limits_{b'\in \cB}\left|\cP^{(a',b')}_{s's_\tau}v(s')-\cP^{(a',b')}_{s's_\tau}v'(s')\right|
\\&\leq \gamma\left\|v-v'\right\|.
\end{align*}
where we have again used the fact that for any scalars $a,b,c$ we have that $
    \left|\max\{a,b\}-\max\{b,c\}\right|\leq \left|a-c\right|$ using the non-expansiveness of $\cP$.
    
Case 2:
\begin{align}\nonumber
&\max\limits_{a'\in \cA} \left(\cR(s_\tau,a',b)-c(s_\tau,a')+\gamma\cP^{(a',b)}_{s's_\tau}v(s')\right)
\\&\qquad\qquad\qquad\quad-\underset{b'\in \cB}{\min} \left(\cR(s_\tau,a,b')+c(s_\tau,b')+\gamma\cP^{(a,b')}_{s's_\tau}v'(s')\right)> 0, \quad \forall \bma=(a,b)\in\cA\times\cB.
\end{align}
Recall that the cost function is bounded in the following way $0<c(s,y)<\|c\|_\infty$ for any $y\in \cA\cup\cB$. Using this we have that $c(s,b)>-c(s,a)$ for any $a\in\cA, b\in\cB$ hence,
\begin{align*}
&(\bm\cM_1 \bmQ-\bm\cM_2 \bmQ')(s_{\tau},\boldsymbol{a})
\\&=    \underset{a'\in \cA}{\max} \left(\cR(s_\tau,a',b)-c(s_\tau,a')+\gamma\cP^{(a',b)}_{s's_\tau}v(s')\right)-\underset{b'\in \cB}{\min} \left(\cR(s_\tau,a,b')+c(s,b')+\gamma\cP^{(a,b')}_{s's_\tau}v'(s')\right)
\\&\leq    \underset{a'\in \cA}{\max} \left(\cR(s_\tau,a',b)+c(s_\tau,b)+\gamma\cP^{a',b}_{s's_\tau}v(s')\right)-\underset{b'\in \cB}{\min} \left(\cR(s_\tau,a,b')+c(s_\tau,b')+\gamma\cP^{(a,b')}_{s's_\tau}v'(s')\right) 
\\&\begin{aligned}\leq \Bigg|\underset{a'\in \cA}{\max}\underset{b'\in \cB}{\max} &\left(\cR(s_\tau,a',b)+c(s,b')+\gamma\cP^{(a',b')}_{s's_\tau}v(s)\right)
\\&\qquad\qquad\qquad\qquad\qquad\qquad\qquad\quad-\underset{a'\in \cA}{\min}\underset{b'\in \cB}{\min} \left(\cR(s_\tau,a',b)+c(s,b')+\gamma\cP^{(a',b')}_{s's_\tau}v'(s')\right) \Bigg|
\end{aligned}
\\&\begin{aligned}
\leq \Bigg|\underset{a'\in \cA}{\max}\underset{b'\in \cB}{\max}& \left(\cR(s_\tau,a',b')+c(s,b')+\gamma\cP^{(a',b')}_{s's_\tau}v(s')\right)
\\&\qquad\qquad\qquad\qquad\qquad\quad\qquad\quad+\underset{a'\in \cA}{\max}\underset{b'\in \cB}{\max} \left(-\cR(s_\tau,a',b')-c(s,b')-\gamma\cP_{s's_\tau}^{(a',b')}v'(s')\right) \Bigg|
\end{aligned}
\\&\leq \left|\underset{a'\in \cA}{\max}\underset{b'\in \cB}{\max} \left(\gamma\cP^{(a',b')}_{s's_\tau}v(s')\right)-\gamma\cP_{s's_\tau}^{(a',b')}v'(s') \right|
\\&\leq \gamma\underset{a'\in \cA}{\max}\underset{b'\in \cB}{\max}\left| \cP_{s's_\tau}^{(a',b')}(v-v')(s')\right| 
\\&\leq \gamma\left\|v-v'\right\| 
\end{align*}
where we have used the Cauchy-Schwarz inequality in the penultimate step.

For the reverse inequality, we have
\begin{align*}
&(\bm\cM_2 \bmQ-\bm\cM_1 \bmQ_1)(s_{\tau},\boldsymbol{a})
\\&\hspace{-1cm}=    \underset{b'\in \cB}{\min} \left(\cR(s_\tau,a,b')+c(s,b')+\gamma\cP^{(a,b')}_{s's_\tau}v'(s')\right)-\underset{a'\in \cA}{\max} \left(\cR(s_\tau,a',b)-c(s_\tau,a')+\gamma\cP^{(a',b)}_{s's_\tau}v(s')\right)
\\&\hspace{-1cm}\begin{aligned}\leq  \max\left\{-\underset{a'\in \cA}{\max} \left(\cR(s_\tau,a',b)-c(s_\tau,a')+\gamma\cP^{(a',b)}_{s's_\tau}v(s')\right),-\underset{b'\in \cB}{\min} \left(\cR(s_\tau,a,b')+c(s_\tau,b')+\gamma\cP^{(a,b')}_{s's_\tau}v(s')\right)\right\}&
\\+\underset{b'\in \cB}{\min} \left(\cR(s_\tau,a,b')+c(s,b')+\gamma\cP^{(a,b')}_{s's_\tau}v'(s')\right)&
\end{aligned}
\\&\hspace{-1cm}\begin{aligned}\leq  \max\left\{-\underset{a'\in \cA}{\max} \left(\cR(s_\tau,a',b)-c(s_\tau,a')+\gamma\cP^{(a',b)}_{s's_\tau}v(s')\right),-\underset{b'\in \cB}{\min} \left(\cR(s_\tau,a,b')+c(s_\tau,b')+\gamma\cP^{(a,b')}_{s's_\tau}v(s')\right)\right\}&
\\-\max\left\{-\underset{a'\in \cA}{\max} \left(\cR(s_\tau,a',b)-c(s_\tau,a')+\gamma\cP^{(a',b)}_{s's_\tau}v(s')\right),-\underset{b'\in \cB}{\min} \left(\cR(s_\tau,a,b')+c(s_\tau,b')+\gamma\cP^{(a,b')}_{s's_\tau}v'(s')\right)\right\}&
\\+\max\left\{-\underset{a'\in \cA}{\max} \left(\cR(s_\tau,a',b)-c(s_\tau,a')+\gamma\cP^{(a',b)}_{s's_\tau}v(s')\right),-\underset{b'\in \cB}{\min} \left(\cR(s_\tau,a,b')+c(s_\tau,b')+\gamma\cP^{(a,b')}_{s's_\tau}v'(s')\right)\right\}&
\\+\underset{b'\in \cB}{\min} \left(\cR(s_\tau,a,b')+c(s,b')+\gamma\cP^{(a,b')}_{s's_\tau}v'(s')\right)&
\end{aligned}
\\&\hspace{-1cm}\begin{aligned}\leq  \Bigg|\max\left\{-\underset{a'\in \cA}{\max} \left(\cR(s_\tau,a',b)-c(s_\tau,a')+\gamma\cP^{(a',b)}_{s's_\tau}v(s')\right),-\underset{b'\in \cB}{\min} \left(\cR(s_\tau,a,b')+c(s_\tau,b')+\gamma\cP^{(a,b')}_{s's_\tau}v(s')\right)\right\}&
\\-\max\left\{-\underset{a'\in \cA}{\max} \left(\cR(s_\tau,a',b)-c(s_\tau,a')+\gamma\cP^{(a',b)}_{s's_\tau}v(s')\right),-\underset{b'\in \cB}{\min} \left(\cR(s_\tau,a,b')+c(s_\tau,b')+\gamma\cP^{(a,b')}_{s's_\tau}v'(s')\right)\right\}\Bigg|&
\\+\Bigg|\max\left\{-\underset{a'\in \cA}{\max} \left(\cR(s_\tau,a',b)-c(s_\tau,a')+\gamma\cP^{(a',b)}_{s's_\tau}v(s')\right),-\underset{b'\in \cB}{\min} \left(\cR(s_\tau,a,b')+c(s_\tau,b')+\gamma\cP^{(a,b')}_{s's_\tau}v'(s')\right)\right\}&
\\+\underset{b'\in \cB}{\min} \left(\cR(s_\tau,a,b')+c(s,b')+\gamma\cP^{(a,b')}_{s's_\tau}v'(s')\right)\Bigg|&
\end{aligned}
\\&\hspace{-2cm}\begin{aligned}\leq  \left|\underset{b'\in \cB}{\max} \left[\cR(s_\tau,a,b')+c(s_\tau,b')+\gamma\cP^{(a,b')}_{s's_\tau}v(s')-
\left(\cR(s_\tau,a,b')+c(s_\tau,b')+\gamma\cP^{(a,b')}_{s's_\tau}v'(s')\right)\right]\right|&
\\+\Bigg|\max\left\{\underset{b'\in \cB}{\min} \left(\cR(s_\tau,a,b')+c(s_\tau,b')+\gamma\cP^{(a,b')}_{s's_\tau}v'(s')\right)-\underset{a'\in \cA}{\max} \left(\cR(s_\tau,a',b)-c(s_\tau,a')+\gamma\cP^{(a',b)}_{s's_\tau}v(s')\right),0\right\}&
\end{aligned}
\\&\hspace{-1cm}\leq  \gamma\max\limits_{a'\in \cA}\max\limits_{b'\in \cB}\left|\left[\cP_{s's_\tau}^{(a',b')}\left(v(s')-v'(s')\right)\right]\right|
\\&\hspace{-1cm}\leq  \gamma\left\|v-v'\right\|
\end{align*}
where we have again used the fact that for any scalars $a,b,c$ we have that $
    \left|\max\{a,b\}-\max\{b,c\}\right|\leq \left|a-c\right|$ using the non-expansiveness of $\cP$.

We now prove iv). We again split the proof of the statement into two cases:

\textbf{Case 1:} 
\begin{align}\bm\cM_i\bm{Q}(s_{\tau},\boldsymbol{a})-\left(\cR(s_\tau,\bm0)+\gamma\mathcal{P}^{\boldsymbol{0}}_{s's_\tau}v'(s')\right)<0.
\end{align}

We now observe the following:
\begin{align*}
&\bm\cM_i\bm{Q}(s_{\tau},\boldsymbol{a})-\left(\cR(s_\tau,\boldsymbol{0})+\gamma\mathcal{P}^{\boldsymbol{0}}_{s's_\tau}v'(s')\right)
\\&\leq\max\left\{\cR(s_\tau,\boldsymbol{0})+\gamma\cP_{s's_\tau}^{\bm0}v(s'),\bm\cM_i\bm{Q}(s_{\tau},\boldsymbol{a})\right\}
-\left(\cR(s_\tau,\boldsymbol{0})+\gamma\mathcal{P}^{\boldsymbol{0}}_{s's_\tau}v'(s')\right)
\\&\leq \Bigg|\max\left\{\cR(s_\tau,\boldsymbol{0})+\gamma\cP_{s's_\tau}^{\bm0}v(s'),\bm\cM_i\bm{Q}(s_{\tau},\boldsymbol{a})\right\}
-\max\left\{\cR(s_\tau,\boldsymbol{0})+\gamma\mathcal{P}^{\boldsymbol{0}}_{s's_\tau}v'(s'),\bm\cM_i\bm{Q}(s_{\tau},\boldsymbol{a})\right\}
\\&\qquad+\max\left\{\cR(s_\tau,\boldsymbol{0})+\gamma\mathcal{P}^{\boldsymbol{0}}_{s's_\tau}v'(s'),\bm\cM_i\bm{Q}(s_{\tau},\boldsymbol{a})\right\}
-\left(\cR(s_\tau,\boldsymbol{0})+\gamma\mathcal{P}^{\boldsymbol{0}}_{s's_\tau}v'(s')\right)\Bigg|
\\&\leq \Bigg|\max\left\{\cR(s_\tau,\boldsymbol{0})+\gamma\mathcal{P}^{\boldsymbol{0}}_{s's_\tau}v(s'),\bm\cM_i\bm{Q}(s_{\tau},\boldsymbol{a})\right\}
-\max\left\{\cR(s_\tau,\boldsymbol{0})+\gamma\mathcal{P}^{\boldsymbol{0}}_{s's_\tau}v'(s'),\bm\cM_i\bm{Q}(s_{\tau},\boldsymbol{a})\right\}\Bigg|
\\&\qquad+\Bigg|\max\left\{\cR(s_\tau,\boldsymbol{0})+\gamma\mathcal{P}^{\boldsymbol{0}}_{s's_\tau}v'(s'),\bm\cM_i\bm{Q}(s_{\tau},\boldsymbol{a})\right\}-\left(\cR(s_\tau,\boldsymbol{0})+\gamma\mathcal{P}^{\boldsymbol{0}}_{s's_\tau}v'(s')\right)\Bigg|
\\&\leq \gamma\max\limits_{\bma\in\bm\cA}\left|\cP_{s's_\tau}^{\bma}v(s')-\cP_{s's_\tau}^{\bma}v'(s')\right|+\left|\max\left\{0,\bm\cM_i\bm{Q}(s_{\tau},\boldsymbol{a})-\left(\cR(s_\tau,\boldsymbol{0})+\gamma\mathcal{P}^{\boldsymbol{0}}_{s's_\tau}v'(s')\right)\right\}\right|
\\&\leq \max\limits_{\bma\in\bm\cA}\left\|\cP_{s's_\tau}^{\bma}\right\|\left\|v-v'\right\|
\\&\leq \gamma\|v-v'\|,
\end{align*}
where we have again used the fact that for any scalars $a,b,c$ we have that $
    \left|\max\{a,b\}-\max\{b,c\}\right|\leq \left|a-c\right|$ and the non-expansiveness of the $\cP$ operators.

\textbf{Case 2: }
\begin{align*}\bm\cM_i\bm{Q}(s_{\tau},\boldsymbol{a})-\left(\cR(s_\tau,\boldsymbol{0})+\gamma\mathcal{P}^{\boldsymbol{0}}_{s's_\tau}v'(s')\right)\geq 0.
\end{align*}

For this case, first recall that for any $\tau\in\mathcal{F}$ and for any $y\in \cA\cup\cB$ we have $c(s_\tau,y)>\lambda$ for some $\lambda >0$.
\begin{align*}
&\bm\cM_i\bm{Q}(s_{\tau},\boldsymbol{a})-\left(\cR(s_\tau,\boldsymbol{0})+\gamma\mathcal{P}^{\boldsymbol{0}}_{s's_\tau}v'(s')\right)
\\&\leq \bm\cM_i\bm{Q}(s_{\tau},\boldsymbol{a})-\max\limits_{a'\in\cA}\left(\cR(s_\tau,\boldsymbol{0})+\gamma\mathcal{P}^{\boldsymbol{0}}_{s's_\tau}v'(s')-c(s_\tau,a')\right)
\\&\leq \max\limits_{\bma'\in\bm\cA}\left[\cR(s_\tau,\boldsymbol{a})-c(s_\tau,a')+\gamma\cP^{\bma'}_{s's_\tau}v(s')\right]
\\&\qquad\qquad\qquad\qquad\quad-\max\limits_{a'\in\cA}\left(\cR(s_\tau,\boldsymbol{0})-c(s_\tau,a')+\gamma\mathcal{P}^{\boldsymbol{0}}_{s's_\tau}v'(s')\right)
\\&\leq \gamma\max\limits_{\bma\in\bm\cA}\left\|\cP_{s's_\tau}^{\bma}\left(v(s')-v'(s')\right)\right|
\\&\leq \gamma\left|v(s')-v'(s')\right|
\\&\leq \gamma\left\|v-v'\right\|,
\end{align*}
again using the fact that $P$ is non-expansive. Hence we have succeeded in showing that for any $v\in L_2$ we have that
\begin{align}
    \left\|\bm\cM\bmQ- \left(\cR(\cdot,\bm0)+\gamma\mathcal{P}^{\bm0}v'\right)\right\|\leq \gamma\left\|v-v'\right\|.\label{off_M_bound_gen}
\end{align}
Gathering the results of the three cases gives the desired result. 
\end{proof}

\begin{proposition}\label{prop:uniqueness}
The value of the game $\cG$ is unique.    
\end{proposition}
\begin{proof}
By Proposition \ref{lemma:bellman_contraction}, we have that
\begin{align}
 \|T^{k+1}v_S-T^{k}v_S\|\leq\gamma\|T^kv_S-T^{k-1}v_S\|\leq \cdots\leq \gamma^k\|Tv_S-v_S\|.   
\end{align}
The result follows after considering the limit as $k\to \infty$ and using the boundedness of $\|Tv_S-v_S\|$, we deduce that $T^kv_S,T^{k+1}v_S,\ldots,$ is a Cauchy sequence which concludes the proof.
\end{proof}

\begin{proof}[Proof of Prop \ref{prop:uniqueness}]
Suppose there exist two values of the game $\cG$, $v'_S$ and $v_S$. Then, since each is a solution, we have by Proposition \ref{lemma:bellman_contraction} that each is a fixed point of the Bellman operator and hence $Tv_S=v_S$ and $Tv'_S=v'_S$. Hence, we have the following:
\begin{align}
  \|v_S-v'_S\|=\|Tv_S-Tv'_S|\leq \gamma\|v_S-v'_S\|,  
\end{align}
whereafter, we immediately deduce that $v_S=v'_S$.
\end{proof}
Summing up the above results we have succeeded in proving Theorem \ref{thrm:minimax-exist}.

\section*{Proof of Proposition \ref{prop:switching_times}}
\begin{proof}
We do the proof for Player 1 with the proof for Player 2 being analogous. First let us recall that the \textit{intervention time} $\tau_k$ is defined recursively $\tau_k=\inf\{t>\tau_{k-1}|s_t\in A,\tau_k\in\mathcal{F}_t\}$ where $A=\bm\cM Q_1(s_t,a)=Q_1(s_t,a)\}$.
The proof is given by establishing a contradiction.  Therefore define by $\tau_k=\inf\{\cF\ni \tau>\tau_{k-1};\bm\cM_1Q_1(s_\tau,a)= Q_1(s_\tau,a)\}$ and suppose that the intervention time $\tau'_1>\tau_1$ is an optimal intervention time. Construct the strategy $\sigma'\in\Sigma^1$ and $\tilde{\sigma}\in\Sigma^1$ intervention times by $(\tau'_0,\tau'_1,\ldots,)$ and $(\tau'_0,\tau_1,\ldots)$ respectively.  Define by $m=\sup\{t;t<\tau'_1\}$.
By construction, we have that
\begin{align*}
& \quad v^{\bm\sigma'}(s)
\\&=\mathbb{E}\left[\cR(s_{0},\bma_{0})+\mathbb{E}\left[\ldots+\gamma^{l-1}\mathbb{E}\left[\cR(s_{\tau_1-1},\bma_{\tau_1-1})+\ldots+\gamma^{m-l}\mathbb{E}\left[ \bm\cM_1Q^{\bm\sigma'}(s_{\tau'_1},a,0)\right]\right]\right]\right]
\\&=\mathbb{E}\left[\cR(s_{0},\bma_{0})+\mathbb{E}\left[\ldots+\gamma^{l}\mathbb{E}\left[\bmQ^{\bm\sigma'}(s_{\tau_1},a,0)\right]\right]\right]
\\&=\mathbb{E}\left[\cR(s_{0},\bma_{0})+\mathbb{E}\left[\ldots+\gamma^{l}\mathbb{E}\left[\bm\cM_1\bmQ^{\bm\sigma'}(s_{\tau_1},a,0)\right]\right]\right]
\end{align*}
We now use the following observation 
\begin{align}
&\mathbb{E}\left[\bm\cM_1\bmQ^{\bm\sigma'}(s_{\tau_1},a,0)\right]
\\&\leq \max\left\{\bm\cM_1\bmQ^{\bm\sigma'}(s_{\tau_1},a,0),\underset{a_{\tau_1}\in\mathcal{A}}{\max}\;\left[ \cR(s_{\tau_{1}},a_{\tau_{1}},0)+\gamma\sum_{s'\in\mathcal{S}}P(s';a_{\tau_1},0,s_{\tau_1})v^{\bm\sigma}(s')\right]\right\}.
\end{align}

Using this we deduce that
\begin{align*}
&v^{\bm\sigma'}(s)\leq\mathbb{E}\Bigg[\cR(s_{0},\bma_{0})+\mathbb{E}\Bigg[\ldots
\\&+\gamma^{l-1}\mathbb{E}\left[ \cR(s_{\tau_1-1},\bma_{\tau_1-1})+\gamma\max\left\{\bm\cM_1\bmQ^{\bm\sigma'}(s_{\tau_1},a,0),\underset{a_{\tau_1}\in\mathcal{A}}{\max}\;\left[ \cR(s_{\tau_1},a_{\tau_{1}},0)+\gamma\sum_{s'\in\mathcal{S}}P(s';a_{\tau_1},s_{\tau_1})v^{\bm\sigma}(s')\right]\right\}\right]\Bigg]\Bigg]
\\&=\mathbb{E}\left[\cR(s_{0},\bma_{0})+\mathbb{E}\left[\ldots+\gamma^{l-1}\mathbb{E}\left[ \cR(s_{\tau_1-1},\bma_{\tau_1-1})+\gamma\left[T v^{\boldsymbol{\tilde{\sigma}}}\right](s_{\tau_1})\right]\right]\right]=v^{\boldsymbol{\tilde{\sigma}}}(s),
\end{align*}
where the first inequality is true by assumption on $\bm\cM$. This is a contradiction since $\tilde{\sigma}$ is a suboptimal policy for Player $1$. 
Moreover, by invoking the same reasoning, we can conclude that it must be the case that $(\tau_0,\tau_1,\ldots,\tau_{k-1},\tau_k,\tau_{k+1},\ldots,)$ are the optimal intervention times. 
\end{proof}

\section{Proof of Theorem \ref{thm:maxmin_budget}}
\begin{proof}
The proof of the Theorem is straightforward since by Theorem \ref{theorem:joint-sol}, each player's budgeted problem can be solved using a dynamic programming principle extension of the Bellman equation corresponding to \eqref{bellman_op}. The result follows by application of Theorem 2 in \citep{sootla2022saute} with minor modifications.

\end{proof}

\end{document}